\documentclass[
    aps, prl, reprint, superscriptaddress, amsfonts,
]{revtex4-2}

\usepackage{hyperref}
\usepackage[dvipsnames]{xcolor}
\usepackage[export]{adjustbox}
\usepackage{mathtools}
\usepackage{physics}
\usepackage{bm}
\usepackage{bbm}
\usepackage{amsthm}
\usepackage{wrapfig}

\makeatletter
\let\DOTSI\relax
\newcommand*{\Sint}{%
  \DOTSI
  \mathop{%
    \mathpalette\@LetterOnInt{\Sigma}%
  }%
  \mkern-\thinmuskip
  \int
}
\newcommand*{\@LetterOnInt}[2]{%
  \sbox0{$#1\int\m@th$}%
  \sbox2{$%
    \ifx#1\displaystyle
      \textstyle
    \else
      \scriptscriptstyle
    \fi
    #2%
  \m@th$}%
  \dimen@=.4\dimexpr\ht0+\dp0\relax
  \ifdim\dimexpr\ht2+\dp2\relax>\dimen@
    \sbox2{\resizebox*{!}{\dimen@}{\unhcopy2}}%
  \fi
  \dimen@=\wd0 %
  \ifdim\wd2>\dimen@
    \dimen@=\wd2 %
  \fi
  \rlap{\hbox to \dimen@{\hfil
    $#1\vcenter{\copy2}\m@th$%
  \hfil}}%
  \ifdim\dimen@>\wd0 %
    \kern.5\dimexpr\dimen@-\wd0\relax
  \fi
}
\makeatother

\newtheorem{theorem}{Theorem}[section]
\newtheorem{lemma}[theorem]{Lemma}
\newtheorem{corollary}[theorem]{Corollary}

\begin{document}

\title{Exact non-Markovian master equations: a generalized derivation for Gaussian systems}

\author{Antonio D'Abbruzzo}
\email{antonio.dabbruzzo@sns.it}
\affiliation{Scuola Normale Superiore, Pisa, Italy}

\author{Vittorio Giovannetti}
\email{vittorio.giovannetti@sns.it}
\affiliation{Scuola Normale Superiore, Pisa, Italy}
\affiliation{NEST and Istituto Nanoscienze-CNR, Pisa, Italy}

\author{Vasco Cavina}
\email{vasco.cavina@sns.it}
\affiliation{Scuola Normale Superiore, Pisa, Italy}

\date{\today}

\begin{abstract}
    We derive an exact master equation that captures the dynamics of a quadratic quantum system linearly coupled to a Gaussian environment of the same statistics: the Gaussian Master Equation (GME).
    Unlike previous approaches, our formulation applies universally to both bosonic and fermionic setups, and remains valid even in the presence of initial system-environment correlations, allowing for the exact computation of the system's reduced density matrix across all parameter regimes.
    Remarkably, the GME shares the same operatorial structure as the Redfield equation and depends on a single kernel---a dressed environment correlation function accounting for all virtual interactions between the system and the environment.
    This simple structure grants a clear physical interpretation and makes the GME easy to simulate numerically, as we show by applying it to an open system based on two fermions coupled via superconductive pairing.
\end{abstract}

\maketitle


\textit{Introduction.}---Finding convenient ways to unveil the dynamics of open quantum systems~\cite{Breuer2007Book} starting from an effective microscopic description is arguably one of the most important problems in contemporary quantum physics.
In the Markovian regime, where a continuous flow of information from the system to the environment is assumed, the well-known LGKS theorem~\cite{Lindblad1976Classic,Gorini1976Classic} provides the general structure of the master equation and the question of linking it to microscopic parameters has been thoroughly studied~\cite{Davies1974Markovian,Dumcke1979Proper,Rivas2010Markovian}.
However, such characterization is missing in the non-Markovian regime and a heterogeneous variety of approaches has been proposed over the years to still find useful predictions in practical scenarios~\cite{Breuer2016Review,deVega2017Review,Li2020Perspective}.
Such approaches are mostly designed with numerics in mind and they smartly bypass the problem of writing down explicitly the master equation for the system.

The situation is different if one assumes that the system has a quadratic free Hamiltonian while being linearly coupled to a Gaussian environment, which is completely characterized by one- and two-point correlation functions only~\cite{Serafini2023Book,Surace2022FermionGauss}.
In this case, exact master equations can be found in the literature~\cite{Hu1992Classic,Zhang2012ME,Yang2013ME,Yang2016Review,Ferialdi2016ME}, but they are typically restricted to excitation-preserving interactions (with notable specific exceptions, such as the Hu-Paz-Zhang equation~\cite{Hu1992Classic}).
When this is the result of a rotating-wave approximation, it is however known that inaccuracies can be found when trying to capture energy shifts and non-Markovian effects~\cite{Fleming2010RWA,Makela2013RWA,Wu2017FermionRWA}.
A master equation that is free of this assumption has been derived in Ref.~\cite{Ferialdi2016ME} (generalizing Ref.~\cite{Hu1992Classic}), but with few caveats.
First, it is only valid for bosonic particles; an extension to fermionic systems is needed to study, e.g., systems with superconducting pairing~\cite{Kitaev2001Majorana,Ribeiro2015NonMarkov,DAbbruzzo2021Kitaev} or spin-fermion Hamiltonians~\cite{Wu2017FermionRWA}.
Second, it requires system and environment to be uncorrelated at the initial time, which may not be experimentally feasible in the strong-coupling regime~\cite{deVega2017Review,Grabert1988Init,McCaul2017Partition,Los2024Initial}.

In this Letter, we derive a master equation---the Gaussian Master Equation (GME)---that is uniformly valid for all Gaussian systems, with no need of further assumptions beyond Gaussianity.
Our approach applies to both the bosonic and fermionic cases~\footnote{Note that the master equation does not apply to hybrid quadratic models in which fermions are linearly coupled to bosons. However, these models are not Gaussian, as higher-order correlators cannot be expressed in terms of two-point correlators---see, for instance, the Jaynes-Cummings model~\cite{Breuer2007Book}.}, accommodates arbitrary forms of the system-environment coupling, and remains valid in presence of initial correlations between the system and the environment.
The GME exhibits a remarkably simple and compact structure, being formally identical to the Redfield equation~\cite{Breuer2007Book}.
However, its kernel is not the bare environmental correlation function, but a dressed one that incorporates additional virtual interactions between the system and the environment.
This formal correspondence has a great impact when trying to solve the GME in practice.
This marks a difference from other important approaches in which the dependence of the proposed equation on microscopic parameters is arguably more involved \cite{Ferialdi2014Gaussian,Cirio2022Influence,Huang2024Gaussian}.

The manuscript is organized as follows. In the first two sections we discuss two key steps in the derivation of the GME. The first is the shift of the problem’s complexity from the operator structure to the time-integration domain, achieved through the Schwinger–Keldysh contour technique~\cite{Schwinger1961Classic,Keldysh1964Classic}.
This method, widely used in many-body physics for nonequilibrium Green’s functions~\cite{Kamenev2011Book,Stefanucci2013Book} and more recently in stochastic unravelings and master equations~\cite{Karanikas2024Fermions,Reyes2024Keldysh,Reimer2019Keldysh,Cavina2024Stochastic}, allows us to include initial system–environment correlations without significantly increasing the complexity of the derivation.
The second step is a recursive application of Wick's theorem to the environment and system operators, yielding a series whose terms share the same superoperatorial structure and can be resummed into a Redfield-like equation.
We conclude the Letter by putting the GME to the test and by applying it to a fermionic system with superconducting pairing, a case that is not addressable by previous exact master equations.
This example aims to show how the numerical approach to this equation is simple and accessible to a broad audience.
Further clues on possible strategies to simplify the numerics (e.g., third quantization) are contained in the Supplemental Material~\cite{SM}.


\textit{Contour dynamics.}---Let us consider a quantum system and an environment, both made of noninteracting particles of the same statistical type (i.e., both bosonic or both fermionic).
The total Hamiltonian $H$ is
\begin{equation}
    H = H_0 + V,
    \quad
    H_0 = H_S + H_E.
\end{equation}
Here, $H_0$ is the free term, containing an environment Hamiltonian $H_E$ and a system Hamiltonian $H_S$.
Instead, $V$ is the system-environment interaction term, which we take to be of the general form
\begin{equation} \label{eq:V}
    V = \sum_\alpha A_\alpha \otimes B_\alpha,
\end{equation}
with $A_\alpha$ being a system operator and $B_\alpha$ being an environment operator.
Contrary to common practices, we do not take $A_\alpha$ and $B_\alpha$ necessarily Hermitian, even though they must be defined such that $V$ is Hermitian.
The following calculations can also be immediately generalized to the case where $H$ contains some explicit time dependence.

Given the initial state of the system-environment compound $\rho_{SE}(0)$, the state at subsequent time $t \geq 0$ in interaction picture writes
\begin{equation} \label{eq:rho_total}
    \varrho_{SE}(t) = U_I(t,0) \rho_{SE}(0) U_I^\dag(t,0),
\end{equation}
where $U_I(t,0) = \mathbb{T} \mathrm{exp}(-i \int_0^t \mathcal{V}(\tau) \dd{\tau})$ is the time-ordered evolution operator and $\mathcal{V}(\tau) = U_0^\dag(\tau,0) V U_0(\tau,0)$, with $U_0$ being the free propagator defined by $H_0$.
In this paper, we will work in units such that $\hbar = k_B = 1$.

Suppose now that the exponentials are expanded.
The operators that appear are ordered so that, when going from right to left, the time argument first moves from $t$ to $0$ and then from $0$ to $t$.
Let us then introduce the contour $\gamma(t) \coloneqq \gamma_+(t) \oplus \gamma_-(t)$ in Fig.~\ref{fig:contour}, where $\gamma_\pm(t)$ are, respectively, the tracks $t \to 0$ and $0 \to t$ (called ``backward'' and ``forward'' branches~\cite{Stefanucci2013Book}).
We can then write Eq.~\eqref{eq:rho_total} using a single exponential as follows:
\begin{equation} \label{eq:rho_contour_temp}
    \varrho_{SE}(t) = \mathbb{T}\qty{ \exp[ -i\int_{\gamma(t)} \mathcal{V}(z) \dd{z} ] \rho_{SE}(0) },
\end{equation}
where $\mathcal{V}(z)$ is defined to be equal to $\mathcal{V}(\tau)$ on both branches, and $\mathbb{T}$ works as the usual ordering operator, but defined on the contour $\gamma(t)$ (i.e., operators with argument $z$ that are ``later'' on $\gamma(t)$, are placed to the left, see also Refs.~\cite{Stefanucci2013Book,Cavina2024Stochastic}).

\begin{figure}
    \includegraphics[scale=0.16,valign=t]{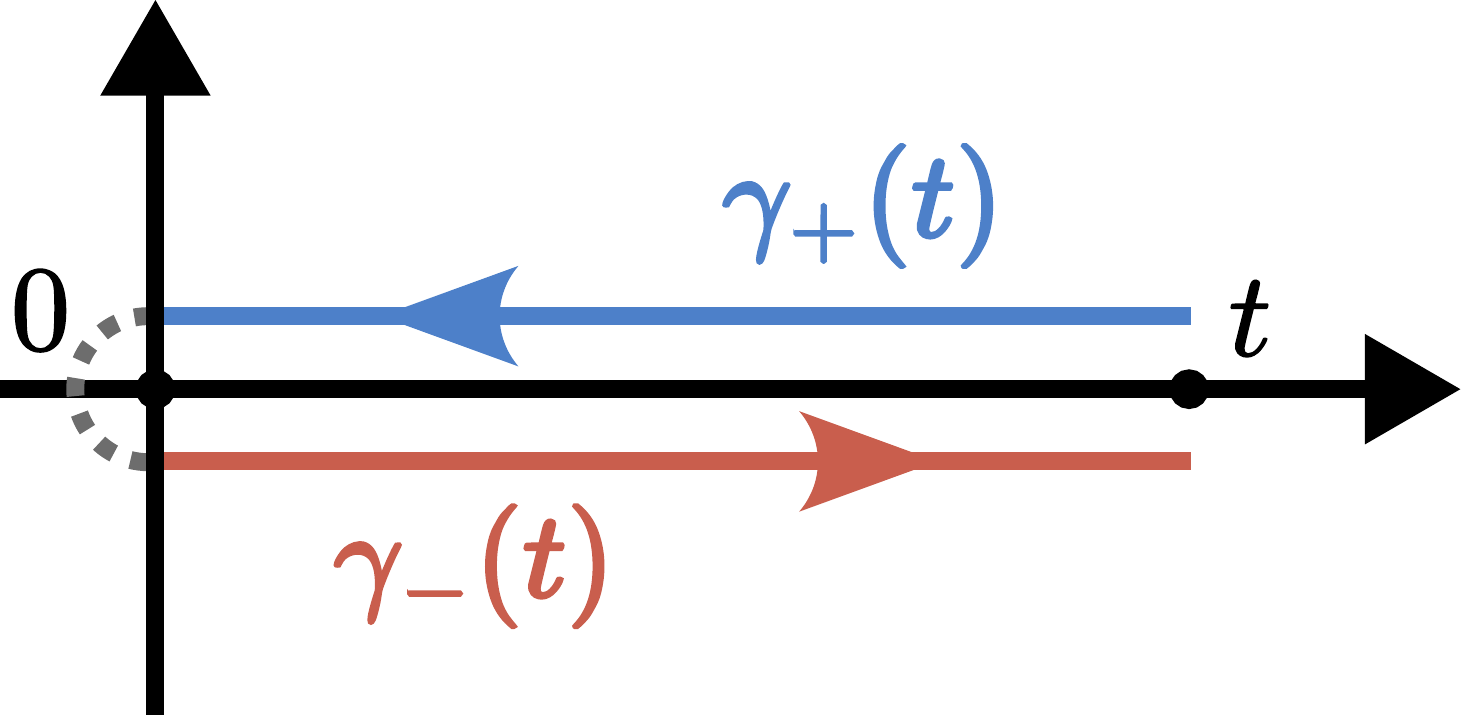}
    \quad
    \includegraphics[scale=0.16,valign=t]{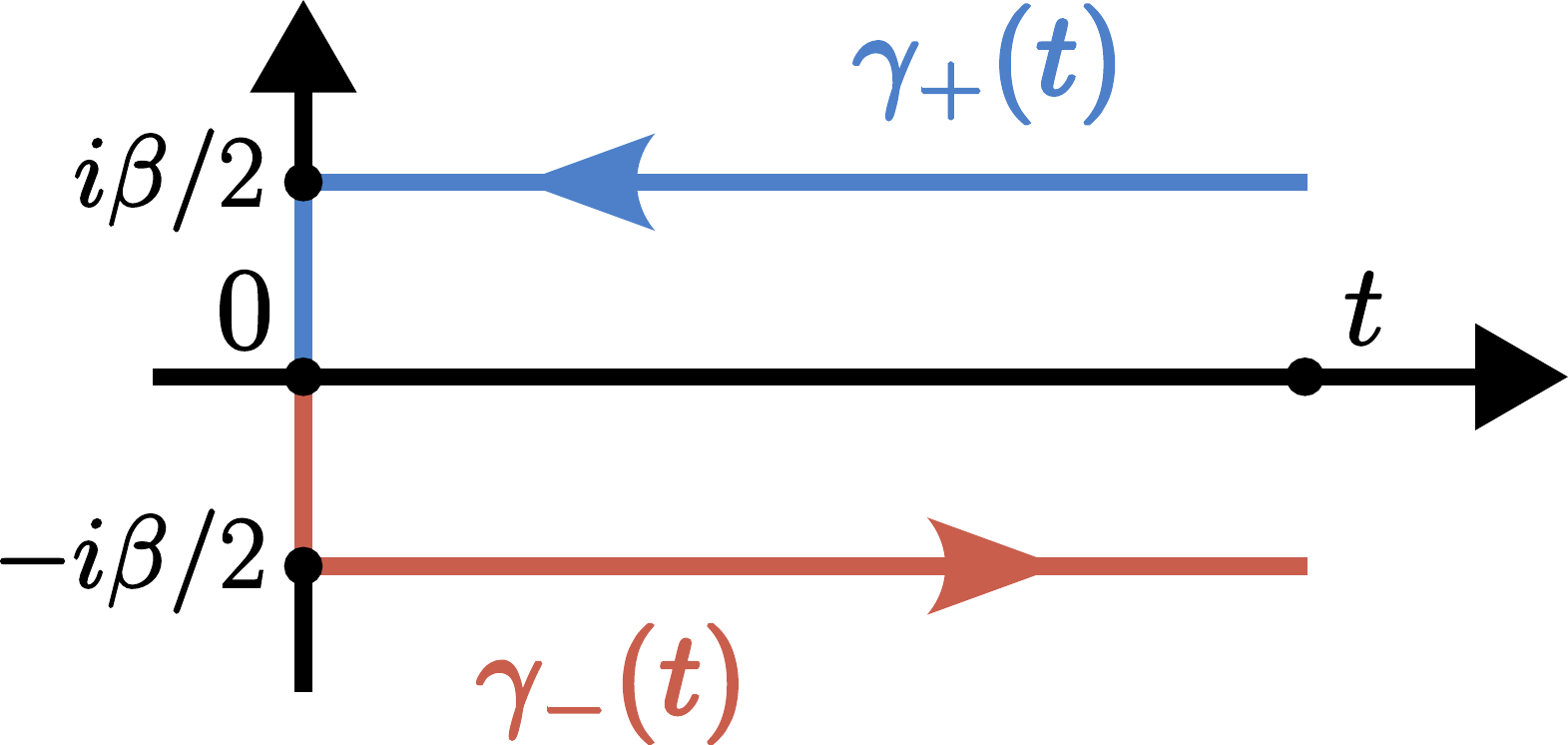}
    \caption{\label{fig:contour}
        Schematic depiction of the contours considered in this paper.
        On the left, $\gamma_+(t)$ goes from $t$ to $0$ while $\gamma_-(t)$ goes from $0$ to $t$: this is used for factorized initial states.
        On the right, $\gamma_\pm(t)$ is shifted by $\pm i \beta/2$ and a vertical track is added that runs from $i\beta/2$ to $-i\beta/2$: this is used for the generic Gaussian initial state in Eq.~\eqref{eq:rho_0}.
    }
\end{figure}

In case $\rho_{SE}(0)$ is factorized, one can immediately take the partial trace over the environment in Eq.~\eqref{eq:rho_contour_temp}, using also the decomposition $\mathcal{V}(z) = \sum_\alpha A_\alpha(z) \otimes B_\alpha(z)$ that comes from Eq.~\eqref{eq:V}.
Unfortunately, this cannot be done for a general $\rho_{SE}(0)$.
However, we can make progress if we assume that $\rho_{SE}(0)$ is a Gaussian state:
\begin{equation} \label{eq:rho_0}
    \rho_{SE}(0) = \frac{e^{-\beta H^M}}{Z_0},
    \qquad
    Z_0 = \Tr[e^{-\beta H^M}],
\end{equation}
where $\beta > 0$ and $H^M$ is a quadratic Hamiltonian in the system-environment variables.
Such initial preparation can be formulated as an imaginary-time evolution, which can then be incorporated into Eq.~\eqref{eq:rho_contour_temp} with a modification of the contour.
Specifically, let us shift $\gamma_\pm(t)$ by $\pm i \beta/2$ and add a vertical track running from $i \beta/2$ to $-i\beta/2$, as illustrated in Fig.~\ref{fig:contour}.
Moreover, divide $H^M = H^M_0 + V^M$ into a free term and an interacting term, and define a new contour interaction $V(z)$ which is equal to $V$ on the horizontal tracks and to $V^M$ on the vertical one.
Then, one can show that~\cite{SM}
\begin{equation}
    \varrho_{SE}(t) = \mathbb{T}\qty{ \exp[-i \int_{\gamma(t)} \mathcal{V}(z) \dd{z}] \frac{e^{-\beta H^M_0}}{Z_0} },
\end{equation}
where this time
\begin{equation} \label{eq:Vz}
    \mathcal{V}(z) = \begin{cases}
        W_0(-i \frac{\beta}{2},z) V(z) W_0(z, -i \frac{\beta}{2}) & z \in \gamma_-(t), \\
        W_0(i \frac{\beta}{2}, z) V(z) W_0(z, i \frac{\beta}{2}) & z \in \gamma_+(t),
    \end{cases}
\end{equation}
and $W_0$ is the contour generalization of the free propagator $U_0$ employing the free contour Hamiltonian $H_0(z)$, which is equal to $H_0$ on the horizontal tracks and to $H_0^M$ on the vertical one.

Writing $\mathcal{V}(z) = \sum_\alpha A_\alpha(z) \otimes B_\alpha(z)$, it is now immediate to take the partial trace over the environment to write the system's state $\varrho(t)$ as
\begin{align}
    \varrho(t) &= \sum_{n=0}^\infty \frac{(-i)^n}{n!} \Sint_{\gamma(t)} \dd[n]{\mathbf z} \Tr[ \mathbb{T}\qty{ B_1 \ldots B_n \Omega_0 } ] \notag \\
    &\times \mathbb{T}\qty{ A_1 \ldots A_n R_0 },
    \label{eq:rho_contour}
\end{align}
where we introduced the abbreviations $A_i \equiv A_{\alpha_i}(z_i)$ and $B_i \equiv B_{\alpha_i}(z_i)$, and the sum symbol superimposed on the integral is a reminder of the fact that we should also sum over the $\alpha$ indices.
We additionally defined
\begin{equation}
    \Omega_0 \coloneqq \frac{e^{-\beta H_E^M}}{\Tr\,[ e^{-\beta H_E^M} ]},
    \quad
    R_0 \coloneqq \frac{e^{-\beta H_S^M}}{Z_0/\Tr\,[ e^{-\beta H_E^M} ]},
\end{equation}
where $H^M_0 = H^M_S + H^M_E$.
Eq.~\eqref{eq:rho_contour} is true even in the fermionic case despite the known issues in dealing with the partial trace in composite fermionic systems, since $\exp[-\beta H^M_0]$ is even in the number of fermionic ladder operators~\cite{Cirio2022Influence}.
Note also that Eq.~\eqref{eq:rho_contour} is structurally identical to what one would obtain by taking the partial trace in Eq.~\eqref{eq:rho_contour_temp} in case of factorized initial state $\rho_{SE}(0) = R_0 \Omega_0$.


\textit{Wick's theorem.}---Since $H_E$ is assumed to be quadratic and since the $B_i$ operators are linear in the environment's ladder operators [the transformation in~\eqref{eq:Vz} does not alter this fact], we can invoke Wick's theorem~\cite{Serafini2023Book,Surace2022FermionGauss,Stefanucci2013Book,Pires2012Wick,Ferialdi2021Wick} to decompose the trace in Eq.~\eqref{eq:rho_contour}.

Let $\zeta \in \{-1,1\}$ be a parameter indicating whether we are dealing with bosons ($\zeta = 1$) or fermions ($\zeta = -1$).
We introduce the symbol $\mathbb{T}_\zeta$: in case $\zeta = 1$ it is equal to $\mathbb{T}$, but in case $\zeta = -1$ it stands for a ``fermionic'' ordering operation, which introduces a $(-1)$ factor every time a transposition is made to order its arguments.
The idea is that the same number of transpositions is required to order $B_1 \ldots B_n \Omega_0$ and $A_1 \ldots A_n R_0$, hence we can harmlessly make the substitution $\mathbb{T} \mapsto \mathbb{T}_\zeta$ in Eq.~\eqref{eq:rho_contour}.
We can then write
\begin{multline} \label{eq:wick}
    \Tr[ \mathbb{T}_\zeta \qty{ B_1 \ldots B_{2m} \Omega_0 } ] \\
    = \frac{1}{m! 2^m} \sum_{\sigma \in \mathfrak{S}_{2m}} \zeta^{N(\sigma)} \mathcal{C}_{\sigma(1),\sigma(2)} \ldots \mathcal{C}_{\sigma(2m-1),\sigma(2m)},
\end{multline}
where $\mathfrak{S}_{2m}$ is the set of permutations of $\{1,\ldots,2m\}$, $N(\sigma)$ is the number of inversions in $\sigma$, and
\begin{equation}
    \mathcal{C}_{i,j} \coloneqq \Tr[ \mathbb{T}_\zeta \qty{ B_i B_j \Omega_0 } ]
\end{equation}
is the environment's correlation function on the contour~\cite{SM}.
Note that Eq.~\eqref{eq:wick} is a straightforward generalization of Wick's theorem on the contour, and no special care is needed, as happens in superoperator formulations~\cite{Cirio2022Influence}.
Wick's theorem also guarantees that $\Tr[ \mathbb{T}_\zeta \qty{ B_1 \ldots B_{2m+1} \Omega_0 } ] = 0$ if we assume as usual that $\Tr[B_i \Omega_0] = 0$, hence we can restrict to the even case.

When substituting Eq.~\eqref{eq:wick} into Eq.~\eqref{eq:rho_contour} we observe with a change of variables that every term of the sum over $\sigma$ contributes in the same way, and that the factor $\zeta^{N(\sigma)}$ is canceled by a reordering on the system side~\cite{SM}.
Therefore, we end up with
\begin{gather}
    \varrho(t) = \sum_{m=0}^\infty \frac{(-1)^m}{m! 2^m} M_m(t), \label{eq:rho_M} \\
    M_m(t) \! \coloneqq \! \Sint_{\gamma(t)} \!\!\! \dd[2m]{\mathbf z} \mathcal{C}_{1,2} \ldots \mathcal{C}_{2m-1,2m} \mathbb{T}_\zeta \qty{ A_1 \ldots A_{2m} R_0 }. \label{eq:M}
\end{gather}


\textit{Exact master equation.}---In order to derive an exact master equation for $\varrho(t)$ we need two tools.
The first arises because of the need to take the time derivative of $\varrho(t)$, and is a ``fundamental theorem of  calculus'' for multidimensional integrals on $\gamma(t)$.
Applying repeatedly the Leibniz integral rule, one finds~\cite{SM}
\begin{equation} \label{eq:derivmany}
    \dv{t} \int_{\gamma(t)} \!\! \dd[n]{\mathbf z} f(\mathbf{z}) =
    n \int_{\gamma(t)} \!\!\! \dd[n-1]{\mathbf z} \qty[ f(t^-,\mathbf{z}) - f(t^+,\mathbf{z}) ],
\end{equation}
where $\tau^\pm$ is the contour point on the horizontal part of $\gamma_\pm(t)$ corresponding to the real time $\tau$, and $f$ is a general contour function, assuming the relative position of the arguments of $f$ to be irrelevant upon integration.

Using appropriate changes of variables and the properties of $\mathbb{T}_\zeta$, it is possible to show that the integrand in Eq.~\eqref{eq:M} satisfies this requirement~\cite{SM}.
We can thus apply Eq.~\eqref{eq:derivmany} to perform the time derivative of $M_m(t)$, resulting in
\begin{multline} \label{eq:M_der}
    \dv{t} M_m(t) = 2m \Sint_{\gamma(t)} \dd[2m-1]{\mathbf z} \mathcal{C}_{\underline{1},2} \mathcal{C}_{3,4} \ldots \mathcal{C}_{2m-1,2m} \\
    \times \qty[ A_{\underline 1}, \mathbb{T}_\zeta \qty{ A_2 \ldots A_{2m} R_0 } ],
\end{multline}
where an underlined index $\underline{i}$ indicates that the corresponding quantity is evaluated in $t^+$ instead of $z_i$.

The second tool we require arises from the need of expressing Eq.~\eqref{eq:M_der} in terms of $M_{k}(t)$ with $k < m$, in order to ``close'' the differential equation.
Let us define the notation $[X,Y]_\zeta \coloneqq XY - \zeta YX$.
Since $H_S$ is quadratic and $A_i$ is linear in the system's ladder operators, $[A_i, A_j]_\zeta$ is a $c$-number for any $i$ and $j$.
This allows us to employ a Wick-like argument to reduce ordered products of system operators to a sum of smaller ordered products~\cite{Ferialdi2021Wick}.
Specifically, consider a product of the form $\mathbb{T}_\zeta \qty{ A_0 A_1 \ldots A_{2m} R_0 }$.
In case $z_0 \in \gamma_-(t)$, $A_0$ appears on the left of $R_0$ after the ordering; we can then (anti)commute it with the $A_i$ operators on its left until $A_0$ is placed at the beginning, leaving behind at each step an ordered product where $A_0$ and $A_i$ are replaced by a scalar contraction proportional to $[A_0, A_i]_\zeta$.
The same can be done in case $z_0 \in \gamma_+(t)$, but this time we (anti)commute $A_0$ with other $A_i$ operators on its right, since it is positioned on the right of $R_0$ by the ordering.
The result is~\cite{SM}
\begin{multline} \label{eq:T_lemma}
    \mathbb{T}_\zeta \qty{ A_0 A_1 \ldots A_{2m} R_0 } = \hat{A}_0 \mathbb{T}_\zeta \qty{A_1 \ldots A_{2m} R_0} \\
    - \sum_{j=1}^{2m} \zeta^{j-1} \Sigma_{0,j} \mathbb{T}_\zeta \qty{ A_1 \ldots A_{j-1} A_{j+1} \ldots A_{2m} R_0 },
\end{multline}
where $\hat{A}_0$ is a superoperator defined as $\hat{A}_0 X = A_0 X$ on $\gamma_-(t)$ and as $\hat{A}_0 X = \zeta X A_0$ on $\gamma_+(t)$, while
\begin{equation}
    \Sigma_{i,j} \coloneqq \begin{cases}
        \theta_{z_i \prec z_j} [A_i, A_j]_\zeta, & z_i \in \gamma_-(t), \\
        - \theta_{z_i \succ z_j} [A_i, A_j]_\zeta, & z_i \in \gamma_+(t),
    \end{cases}
\end{equation}
with $\theta_{z_i \prec z_j}$ being a contour Heaviside step function, that is equal to $1$ if $z_i$ precedes $z_j$ on the contour and $0$ otherwise.
Eq.~\eqref{eq:T_lemma} can be applied repeatedly in Eq.~\eqref{eq:M_der} until we exhaust all the $A_i$ operators.
After a long but straightforward calculation, one concludes that~\cite{SM}
\begin{equation}
    \dv{t} M_m(t) = \sum_{k=1}^{m} \frac{(-1)^{k+1}(2m)!!}{(2m-2k)!!} \hat{\mathcal S}_k(t) M_{m-k}(t),
\end{equation}
where $\hat{\mathcal S}_k(t)$ is a superoperator with action
\begin{multline} \label{eq:S}
    \hat{\mathcal S}_k(t) X \coloneqq \Sint_{\gamma(t)} \dd[2k-1]{\mathbf z} \mathcal{C}_{\underline{1},2} \Sigma_{2,3} \mathcal{C}_{3,4} \ldots \\
    \ldots \Sigma_{2k-2,2k-1} \mathcal{C}_{2k-1,2k} [A_{\underline 1}, \hat{A}_{2k} X].
\end{multline}
When used to calculate the time derivative of $\varrho(t)$ [see Eq.~\eqref{eq:rho_M}], we recognize the Cauchy product of two series, so that~\cite{SM}
\begin{equation} \label{eq:ME_raw}
    \dv{t} \varrho(t) = - \sum_{k=1}^\infty \hat{\mathcal S}_k(t) \varrho(t).
\end{equation}

Eq.~\eqref{eq:ME_raw} can be written in a more illuminating form if we isolate the operator dependence in Eq.~\eqref{eq:S}, which only involves two indices.
Specifically, after restoring the full notation for clarity,
\begin{equation} \label{eq:ME}
    \dv{t} \varrho(t) = \sum_{\alpha,\beta} \int_{\gamma(t)} \dd{z} \mathcal{G}_{\alpha\beta}(t^+,z) [\hat{A}_\beta(z) \varrho(t), A_\alpha(t)],
\end{equation}
where $\mathcal{G}_{p,q} = \sum_{k=1}^{\infty} \mathcal{G}^{(k)}_{p,q}$ and
\begin{subequations} \label{eq:Gseries}
    \begin{gather}
        \mathcal{G}^{(1)}_{p,q} = \mathcal{C}_{p,q}, \\
        \mathcal{G}^{(k)}_{p,q} = \Sint_{\gamma(t)} \dd[2]{\mathbf w} \mathcal{G}^{(k-1)}_{p,1} \Sigma_{1,2} \mathcal{C}_{2,q}.
    \end{gather}
\end{subequations}
If we decompose $\gamma(t)$ in its branches~\cite{SM}, we finally obtain the GME
\begin{equation} \label{eq:ME_redfield}
    \begin{split}
        \dv{\varrho(t)}{t} = \sum_{\alpha,\beta} \int_{\gamma_-(t)} \dd{z} \mathcal{G}_{\alpha\beta}(t^+,z) [A_\beta(z) \varrho(t), A_\alpha(t)] + \text{H.c.}
    \end{split}
\end{equation}
Note that the integral over $\gamma_-(t)$ reduces to an integral over $[0,t]$ in case of factorized initial state.
Thus, this is formally a Redfield equation~\cite{Breuer2007Book,deVega2017Review}, with an additional contribution from initial correlations and with the function $\mathcal{G}$ replacing the ``bare'' correlation $\mathcal{C}$.
The physical meaning of this replacement becomes evident if we realize that Eqs.~\eqref{eq:Gseries} can be written as follows after summing over the index $k$:
\begin{equation} \label{eq:dyson}
    \mathcal{G}_{p,q} = \mathcal{C}_{p,q} + \Sint_{\gamma(t)} \dd[2]{\mathbf w} \mathcal{G}_{p,1} \Sigma_{1,2} \mathcal{C}_{2,q}.
\end{equation}
This has the shape of a Dyson equation~\cite{Stefanucci2013Book}, where $\mathcal{C}$ assumes the role of a ``noninteracting'' Green's function and $\Sigma$ assumes the role of a self-energy.
The solution of the Dyson equation $\mathcal{G}$ represents a ``dressed'' correlation function that acts as a memory kernel:
$\mathcal{G}^{(k)}$ can be thought of as carrying the effect on the dynamics of $\varrho(t)$ coming from $k$ past interactions between the system and the environment, and is of order $2k$ in the coupling.
With this picture in mind, the fact that a truncation at $k=1$ yields the ``memoryless'' Redfield equation acquires a neat physical interpretation.
Note that Eq.~\eqref{eq:dyson} can be solved using well-studied tools from many-body and condensed matter physics: see, e.g., the recent works in Refs.~\cite{Talarico2019Dyson,Lamic2024Dyson}.

Once $\mathcal{G}$ is found, the GME~\eqref{eq:ME_redfield} can be solved using techniques borrowed from the literature on third quantization~\cite{Prosen2008Third,DAbbruzzo2021Self,Barthel2022Third}.
Specifically, one can always turn Eq.~\eqref{eq:ME_redfield} into a continuous differential Lyapunov equation for the covariance matrix associated with the system's Gaussian state~\cite{Serafini2023Book,Surace2022FermionGauss,Barthel2022Third,Purkayastha2022Lyapunov,Gajic1995Lyapunov}.
In the Supplemental Material~\cite{SM} we report how the Lyapunov equation looks like specifically for Eq.~\eqref{eq:ME_redfield}.


\textit{Example.}---We conclude by discussing a simple example application.
Consider as a system the two-mode fermionic Hamiltonian
\begin{equation} \label{eq:example_hs}
    H_S = \epsilon_1 a_1^\dag a_1 + \epsilon_2 a_2^\dag a_2 + \delta \qty(a_1^\dag a_2^\dag - a_1 a_2).
\end{equation}
We also couple both modes to a reservoir through standard tunneling:
\begin{gather} \label{eq:example_int}
        H_E = \sum_r \varepsilon_r c_r^\dag c_r, \quad
        V = \sum_r g_r c_r \qty(a_1^\dag + a_2^\dag) + \text{H.c.}
\end{gather}
To keep things simple, we assume Lorentzian spectral density, zero temperature, large negative chemical potential, and initial decoupling from the system, so that the ``bare'' correlation function can be taken as
\begin{equation} \label{eq:c}
    c(t) = \frac{\gamma \lambda}{2} e^{-\lambda |t|},
    \quad
    \gamma,\lambda > 0
\end{equation}
[see~\cite{SM} for the explicit link between $c(t)$ and $\mathcal{C}$].

Because of the fermionic character and the pairing $\delta$, this setup cannot be treated using the existing exact master equations mentioned at the beginning of the paper: however, it is well within the scope of our proposal.
From a physical standpoint, note that $H_S$ can be used as a toy model for two quantum dots coupled to a $p$-wave superconductor, which induces the pairing $\delta$ by proximity effect~\cite{Kitaev2001Majorana}.
With appropriate extensions, this example could serve as a starting point for studying, e.g., the exact dynamics of a $p$-wave Cooper pair splitter~\cite{Soller2012Cooper,Wang2022Singlet}.

\begin{figure}
    \centering
    \includegraphics[width=\columnwidth]{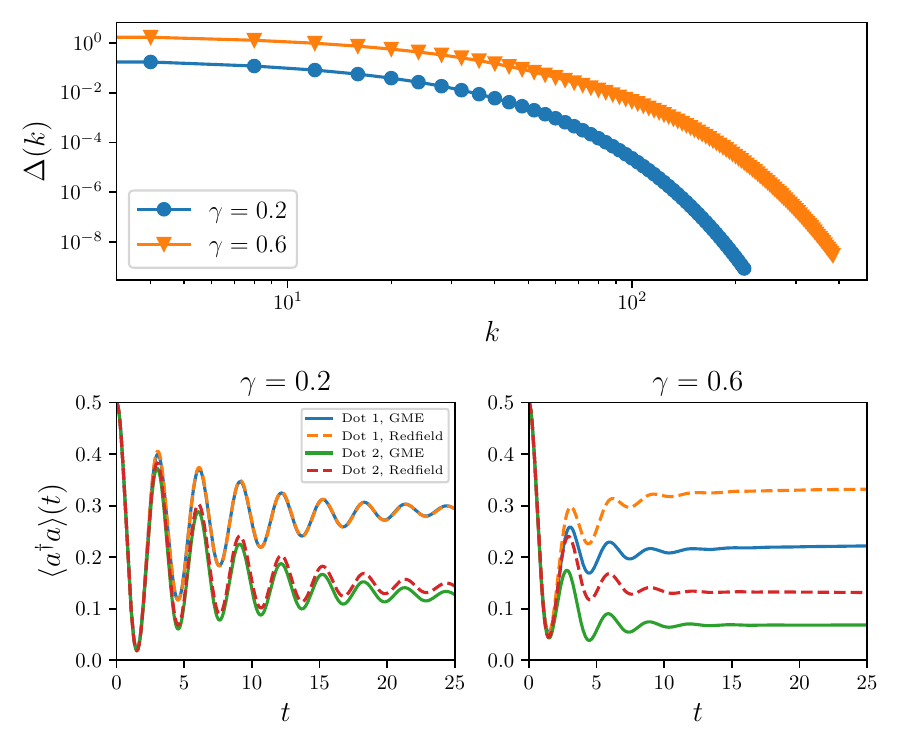}
    \caption{\label{fig:dots_sim}
        (Top) Correction to the solution of the Dyson equation at $k$th order $\Delta(k) = \max \norm*{\mathcal{G}^{(k+1)}}$, maximized over the time grid, for several values of the coupling $\gamma$.
        (Bottom) Evolution of the dots' populations, with $\varepsilon_1 = 0.5$, $\epsilon_2 = 1$, $\delta = 0.7$, and $\lambda = 1.5$.
        The Dyson equation is solved up to the higher order shown in the top plot.
    }
\end{figure}

We approach the problem numerically using the following steps~\cite{SM}.
First, we expand the Dyson equation~\eqref{eq:dyson} in its real-time components and we solve it iteratively on a discrete time grid to find $\mathcal{G}$ up to a certain order.
As shown in Fig.~\ref{fig:dots_sim}, we find that the convergence rate depends on $\gamma$: the stronger the coupling the higher the order one should stop at to achieve a target accuracy.

After this, we are able to solve the GME with a standard Runge-Kutta integrator.
In Fig.~\ref{fig:dots_sim} we report the obtained evolution of the dots' populations $\langle a_i^\dag a_i \rangle(t)$, after having initialized the system in the state $(\ket{00} + \ket{11})/\sqrt{2}$, for two different values of $\gamma$.
The competition between dissipation and pairing causes damped oscillations that eventually settle on a steady-state value.
We also plot the results predicted by the Redfield equation: as expected, relevant deviations from the GME are found at strong coupling.


\textit{Conclusions.}---We derived the GME, an exact master equation valid for all Gaussian systems which is remarkably compact and numerically efficient.
The GME is a promising tool to study the dynamics of a variety of settings that were precluded to other non-Markovian master equations, such as Kitaev chains and superconductive networks~\cite{Kitaev2001Majorana}.
The application to such models is instrumental to understand the role of non-Markovianity in the noise affecting quantum computing architectures.
From a theoretical point of view, the natural next step is to generalize the approach to system-environment pairs obeying different statistics, as this would allow to solve the longstanding problem of finding an exact master equation for the general spin-boson model~\cite{Ferialdi2017SpinBoson}.
It could also be possible to extend the analysis to interacting systems and/or non-Gaussian environments using recursive techniques~\cite{Gasbarri2018Recursive,Trushechkin2021Recursive} or self-consistent mean-field approaches~\cite{Stefanucci2013Book,Scarlatella2024SelfConsistent}.


\textit{Acknowledgments.}---We thank Guglielmo Pellitteri for his valuable suggestions on how to implement the numerical solution of the GME in Fig.~\ref{fig:dots_sim}.
We acknowledge financial support by MUR (Ministero dell'Università e della Ricerca) through the PNRR MUR project PE0000023-NQSTI.


\bibliography{bibliography.bib}


\clearpage
\onecolumngrid
\setcounter{secnumdepth}{2}
\appendix

\begin{center}
    \large{\textbf{Supplemental Material of ``Exact non-Markovian master equations:\\a generalized derivation for Gaussian systems''}}
\end{center}

This Supplemental Material provides a more thorough discussion of some of the statements made in the Letter.
The material presented here is not essential in order to grasp the concepts presented there; however, it is included to aid the interested reader in dwelling into the details.

We start in Sec.~\ref{SM:sec:contour} with a discussion of the shifted contour, showing how the system state can be conveniently written even in the presence of initial system-environment correlations.
Then, in Sec.~\ref{SM:sec:wick} we use the quadratic environment hypothesis and we apply Wick's theorem to the previously obtained expression.
In order to take the time derivative, in Sec.~\ref{SM:sec:calculus} we proceed by proving a ``fundamental theorem of calculus'' for multidimensional contour integrals.
Then, in Sec.~\ref{SM:sec:reduction} we use the quadratic system hypothesis to show how ordered products of system operators can be reduced in terms of smaller products.
This result is used in Sec.~\ref{SM:sec:closing} to close the differential equation for the state, resulting in the exact master equation presented in the Letter.
We also show in Sec.~\ref{SM:sec:physical} how to write it using physical time instead of the abstract contour time, and in Sec.~\ref{SM:sec:covariance} we formulate it as a continuous differential Lyapunov equation for the covariance matrix.
Finally, in Sec.~\ref{SM:sec:example} we provide additional details on the procedure used to simulate the paired double-dot example presented in the Letter.


\section{System state on the contour with initial correlations} \label{SM:sec:contour}

\begin{wrapfigure}{l}{0.28\textwidth}
    \centering
    \includegraphics[scale=0.18]{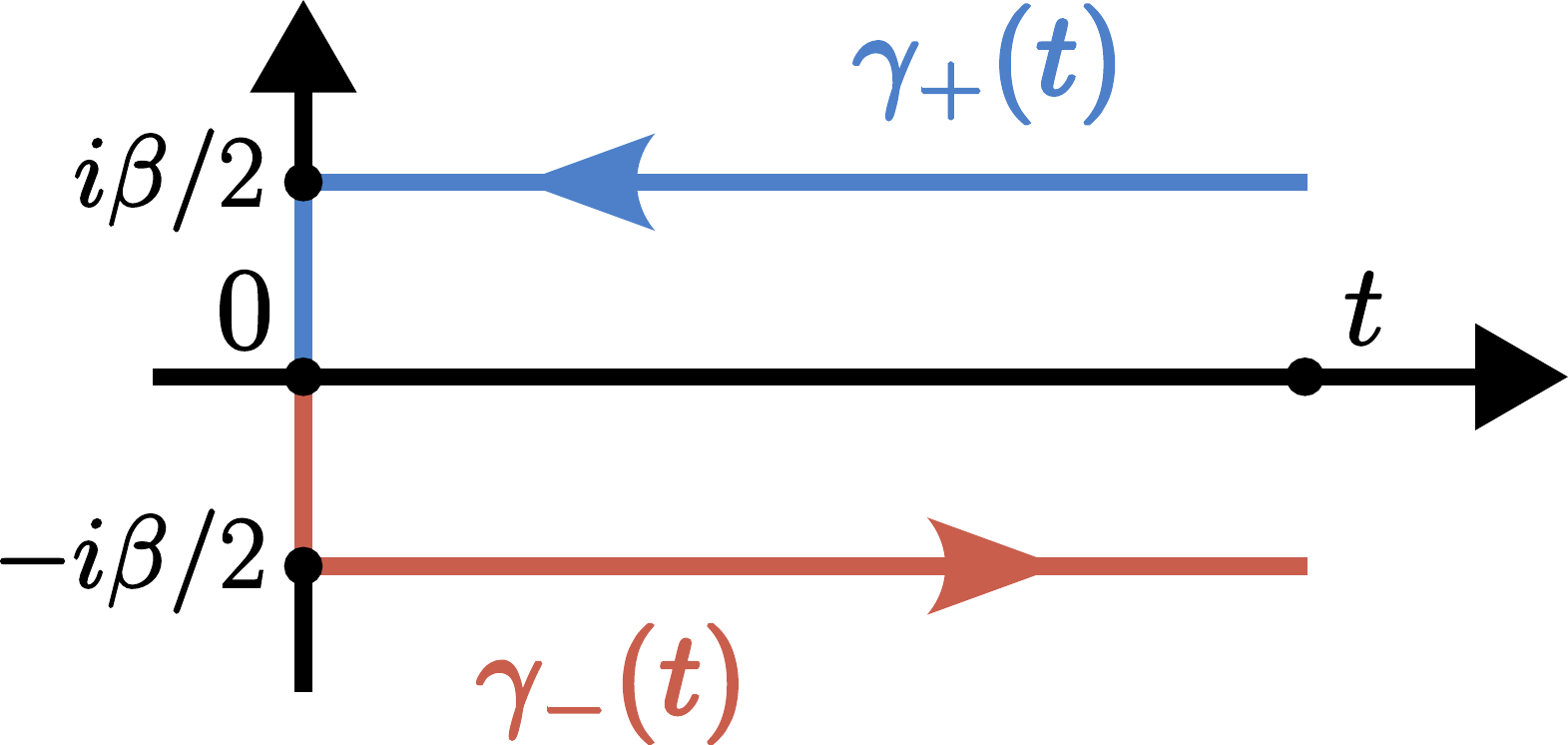}
    \caption{\label{SM:fig:contour}
        Schematic depiction of the shifted contour $\gamma(t)$.
    }
\end{wrapfigure}
The first task is to write the system state using the shifted contour $\gamma(t)$ depicted in Fig.~\ref{SM:fig:contour} when the initial system-environment state is written as
\begin{equation} \label{SM:eq:rho_0}
    \rho_{SE}(0) = \frac{e^{-\beta H^M}}{Z_0},
    \qquad Z_0 = \Tr e^{-\beta H^M},
\end{equation}
with $\beta > 0$ and $H^M = H^M_0 + V^M$ a generic quadratic Hamiltonian.
To do that, we adapt the derivation provided in Ref.~\cite{Cavina2024Stochastic}.

Let us define the following propagator with arguments on $\gamma(t)$:
\begin{equation} \label{SM:eq:contour_propagator}
    W(z_2, z_1) \coloneqq \begin{cases}
        \mathbb{T} \exp( -i \int_{z_1}^{z_2} H(z) \dd{z} ) & z_1 \preceq z_2, \\
        \widetilde{\mathbb T} \exp( i \int_{z_2}^{z_1} H(z) \dd{z} ) & z_1 \succ z_2,
    \end{cases}
\end{equation}
where $\preceq$ is the natural ordering relation on $\gamma(t)$ with corresponding ordering and anti-ordering operations $\mathbb{T}$ and $\widetilde{\mathbb T}$, and $\int_{z_1}^{z_2} \dd{z}$ is a line integral over $\gamma(t)$ starting at $z_1$ and ending at $z_2$.
Here $H(z)$ is defined to be equal to the physical Hamiltonian $H$ on both the horizontal tracks, while it coincides with $H^M$ on the vertical one.
Unlike the standard real-time propagator, $W(z_2,z_1)$ is not unitary (unless $z_1$ and $z_2$ both lie on the same horizontal track).
However, one can easily show that the following properties still hold~\cite{Stefanucci2013Book}:
\begin{subequations}
    \begin{gather}
        W(z,z) = \mathbbm{1}, \\
        W(z_3,z_1) = W(z_3,z_2) W(z_2,z_1), \\
        \dv{W(z_2,z_1)}{z_2} = \begin{cases}
            -i H(z_2) W(z_2, z_1) & z_1 \preceq z_2, \\
            i W(z_2, z_1) H(z_1) & z_1 \succ z_2
        \end{cases} \label{SM:eq:derivative_W}
    \end{gather}
\end{subequations}
We can use $W(z_2,z_1)$ to write the initial system-environment state \eqref{SM:eq:rho_0} as
\begin{equation}
    \rho_{SE}(0) = \frac{e^{-2b H^M}}{Z_0} = \frac{W(-ib,ib)}{Z_0} = \frac{W(-ib,0) W(0,ib)}{Z_0},
\end{equation}
where $b \coloneqq \beta/2$ for brevity.
Note also that line integrals on the horizontal sections do not change after a $\pm i b$ shifting, hence the propagator $U$ defined using $H$ on the physical time satisfies the property $U(t_2, t_1) = W(t_2 \pm ib, t_1 \pm ib)$.
This allows us to write
\begin{equation}
    \rho_{SE}(t) = U(t,0) \rho_{SE}(0) U(0,t) = \frac{1}{Z_0} W(t-i b,-i b) W(-i b,0) W(0,i b) W(i b, t+i b) = \frac{1}{Z_0} W(t-i b,0) W(0,t+i b).
\end{equation}
We can now move to the interaction picture, after introducing a contour propagator $W_0(z_2,z_1)$ defined as in Eq.~\eqref{SM:eq:contour_propagator} but with $H_0(z)$ instead of $H(z)$.
Obviously, $H_0(z)$ is defined to be equal to $H_0$ on the horizontal tracks and to $H_0^M$ on the vertical one.
\begin{equation} \label{SM:eq:rho_interaction_temp}
    \varrho_{SE}(t) = U_0(0,t) \rho_{SE}(t) U_0(t,0) = \frac{1}{Z_0} W_0(-i b,t-i b) W(t-i b,0) W(0,t+i b) W_0(t+i b,i b).
\end{equation}
Here, $U_0$ is the standard propagator defined using $H_0$ on the physical time.
The following lemma is needed to proceed.
\begin{lemma}
    If $z_2 \succeq z_1$ and $z_2 \succeq \overline{z}$, then
    \begin{equation} \label{SM:eq:lemma_W}
        W(z_2,z_1) = W_0(z_2,\overline{z}) W_I(z_2,z_1;\overline{z}) W_0(\overline{z},z_1),
    \end{equation}
    where $W_I(z_2,z_1;\overline{z})$ is defined as in Eq.~\eqref{SM:eq:contour_propagator} but with $W_0(\overline{z},z) V(z) W_0(z,\overline{z})$ instead of $H(z)$. Here $V(z)$ is equal to $V$ on the horizontal tracks and to $V^M$ on the vertical one.
\end{lemma}
\begin{proof}
    It is sufficient to show that the two sides of Eq.~\eqref{SM:eq:lemma_W} satisfy the same differential equation.
    Of course, they are both equal to the identity in case $z_1 = z_2$.
    Moreover,
    \begin{equation}
        \begin{split}
            \dv{z_2} &\qty[W_0(z_2,\overline{z}) W_I(z_2,z_1;\overline{z}) W_0(\overline{z},z_1)] \\
            &= -i H_0(z_2) W_0(z_2,\overline{z}) W_I(z_2,z_1;\overline{z}) W_0(\overline{z},z_1) -i W_0(z_2,\overline{z}) W_0(\overline{z},z_2) V(z) W_0(z_2,\overline{z}) W_I(z_2,z_1;\overline{z}) W_0(\overline{z},z_1) \\
            &= -i H(z_2) \qty[W_0(z_2,\overline{z}) W_I(z_2,z_1;\overline{z}) W_0(\overline{z},z_1)],
        \end{split}
    \end{equation}
    which should be compared with Eq.~\eqref{SM:eq:derivative_W}.
\end{proof}

This lemma allows us to write
\begin{subequations}
    \begin{gather}
        W_0(-i b,t-i b) W(t-i b,0) = W_I(t-i b,0; -i b) W_0(-i b,0), \\
        W(0,t+i b) W_0(t+i b,i b) = W_0(0,i b) W_I(0,t+i b; i b),
    \end{gather}
\end{subequations}
which can be used in Eq.~\eqref{SM:eq:rho_interaction_temp} to arrive at
\begin{equation}
    \varrho_{SE}(t) = \frac{1}{Z_0} W_I(t-i b,0; -i b) W_0(-i b,0) W_0(0,i b) W_I(0,t+i b; i b)
    = W_I(t-i b,0; -i b) \frac{e^{-2b H_0^M}}{Z_0} W_I(0,t+i b; i b).
\end{equation}
Since $W_I(t-ib,0; -ib)$ involves a line integral over $\gamma_-(t)$ and $W_I(0,t+ib; ib)$ involves a line integral over $\gamma_+(t)$, we can arrange the entire expression as a single line integral over $\gamma(t)$ if we include $e^{-2b H_0^M}/Z_0$ in the ordered products:
\begin{equation}
    \varrho_{SE}(t) = \mathbb{T} \qty{ \exp[ -i \int_{\gamma(t)} \mathcal{V}(z) \dd{z} ] \frac{e^{-2b H_0^M}}{Z_0} },
\end{equation}
where $\mathcal{V}(z) \coloneqq W_0(\overline{z},z) V W_0(z,\overline{z})$, with $\overline{z}$ being $\pm ib$ when $z \in \gamma_\pm(t)$.
This is exactly the content of Eqs.~(6)-(7) of the Letter.
Since $e^{-2b H_0^M} = e^{-2b H_S^M} e^{-2b H_E^M}$ is factorized, we can immediately perform the partial trace over the environment once the exponential is expanded, obtaining
\begin{equation} \label{SM:eq:rho_starting}
    \varrho(t) = \sum_{n=0}^\infty \frac{(-i)^n}{n!} \Sint_{\gamma(t)} \dd[n]{\mathbf z} \Tr[ \mathbb{T}\qty{ B_1 \ldots B_n \Omega_0 } ] \mathbb{T}\qty{ A_1 \ldots A_n R_0 },
\end{equation}
with the notations explained in the Letter.
The same formula holds in case the initial system-environment state is factorized as $\rho_{SE}(0) = R_0 \Omega_0$, provided the contour $\gamma(t)$ is not vertically shifted.
Since the same number of transpositions is involved in ordering the strings $B_1 \ldots B_n \Omega_0$ and $A_1 \ldots A_n R_0$, we can make the substitution $\mathbb{T} \mapsto \mathbb{T}_\zeta$, where the latter symbol is introduced above Eq.~(10) of the Letter.
This will be useful in the next section.

In the following we will use the common notation $\tau^\pm$ to indicate a point on the horizontal part of $\gamma_\pm(t)$ with value $\tau$, independently on the presence or not of the shift $\pm i b$ required to keep correlations into account.


\section{Applying Wick's theorem on the contour} \label{SM:sec:wick}

Now we make the assumption of quadratic environment linearly coupled to the system, so that we can invoke Wick's theorem.
Assuming $\Tr[B_i \Omega_0] = 0$ for every $i$, it assumes the form
\begin{subequations} \label{SM:eq:wick}
    \begin{gather}
        \Tr[\mathbb{T}_\zeta \qty{ B_1 \ldots B_{2m+1} \Omega_0 }] = 0, \\
        \Tr[\mathbb{T}_\zeta \qty{ B_1 \ldots B_{2m} \Omega_0 }] = \frac{1}{m! 2^m} \sum_{\sigma \in \mathfrak{S}_{2m}} \zeta^{N(\sigma)} \mathcal{C}_{\sigma(1),\sigma(2)} \ldots \mathcal{C}_{\sigma(2m-1),\sigma(2m)},
    \end{gather}
\end{subequations}
where $\mathfrak{S}_{2m}$ is the set of permutations of $\qty{1,\ldots,2m}$, $N(\sigma)$ is the number of inversions in $\sigma$, and
\begin{equation} \label{SM:eq:C}
    \mathcal{C}_{i,j} \coloneqq \Tr[\mathbb{T}_\zeta \qty{B_i B_j \Omega_0}] = \zeta \mathcal{C}_{j,i}
\end{equation}
is the environment's ``bare'' correlation function on the contour.
Here by ``inversion'' we mean the number of pairs $(i,j)$ such that $i < j$ but $\sigma(i) > \sigma(j)$: in this way, $\zeta^{N(\sigma)}$ is equal to the sign of $\sigma$ in the fermionic scenario.
Using Eqs.~\eqref{SM:eq:wick} in Eq.~\eqref{SM:eq:rho_starting},
\begin{equation}
    \varrho(t) = \sum_{m=0}^{\infty} \frac{(-i)^{2m}}{(2m)! m! 2^m} \sum_{\sigma \in \mathfrak{S}_{2m}} \zeta^{N(\sigma)} \Sint_{\gamma(t)} \dd[n]{\mathbf z} \mathcal{C}_{\sigma(1),\sigma(2)} \ldots \mathcal{C}_{\sigma(2m-1),\sigma(2m)} \mathbb{T}_\zeta \qty{A_1 \ldots A_{2m} R_0}.
\end{equation}
At this point we make the change of variables $w_i = z_{\sigma(i)}$ and $\beta_i = \alpha_{\sigma(i)}$, so that
\begin{equation}
    \varrho(t) = \sum_{m=0}^{\infty} \frac{(-i)^{2m}}{(2m)! m! 2^m} \sum_{\sigma \in \mathfrak{S}_{2m}} \zeta^{N(\sigma)} \Sint_{\gamma(t)} \dd[n]{\mathbf w} \mathcal{C}_{1,2} \ldots \mathcal{C}_{2m-1,2m} \mathbb{T}_\zeta \qty{A_{\sigma^{-1}(1)} \ldots A_{\sigma^{-1}(2m)} R_0}.
\end{equation}
However, one can also easily show the relation
\begin{equation}
    \mathbb{T}_\zeta \qty{ A_{\sigma^{-1}(1)} \ldots A_{\sigma^{-1}(2m)} R_0 } = \zeta^{N(\sigma^{-1})} \mathbb{T}\qty{ A_1 \ldots A_{2m} R_0 }.
\end{equation}
The factor $\zeta^{N(\sigma^{-1})}$ cancels out the factor $\zeta^{N(\sigma)}$, and nothing depends on $\sigma$ anymore.
Hence, the sum over $\sigma$ can be replaced with $(2m)!$ and
\begin{equation}
    \varrho(t) = \sum_{m=0}^{\infty} \frac{(-i)^{2m}}{m! 2^m} \Sint_{\gamma(t)} \dd[n]{\mathbf w} \mathcal{C}_{1,2} \ldots \mathcal{C}_{2m-1,2m} \mathbb{T}_\zeta \qty{A_1 \ldots A_{2m} R_0},
\end{equation}
which is the content of Eqs.~(12)-(13) of the Letter.
We rewrite them here for future convenience:
\begin{equation} \label{SM:eq:rho_after_wick}
    \varrho(t) = \sum_{m=0}^\infty \frac{(-1)^m}{m! 2^m} M_m(t),
    \qquad
    M_m(t) \coloneqq \Sint_{\gamma(t)} \dd[n]{\mathbf z} \mathcal{C}_{1,2} \ldots \mathcal{C}_{2m-1,2m} \mathbb{T}_\zeta \qty{A_1 \ldots A_{2m} R_0}.
\end{equation}
Note that this can also be written in a more suggestive exponential form:
\begin{equation}
    \varrho(t) = \mathbb{T}_\zeta \qty{ \exp[ -\frac{1}{2} \Sint_{\gamma(t)} \dd[2]{\mathbf z} \mathcal{C}_{1,2} A_1 A_2 ] R_0 }.
\end{equation}


\section{Fundamental theorem of calculus on the contour} \label{SM:sec:calculus}

In order to find an exact master equation for $\varrho(t)$ we first need to take the time derivative of Eq.~\eqref{SM:eq:rho_after_wick}.
The following theorem can be used for such purpose.
\begin{theorem}
    Let $f(\mathbf z)$ be an operator-valued function defined on $[\gamma(t)]^n$, with $n \in \mathbb{N}_+$.
    Moreover, given a vector $\mathbf{z} = (z_1, \ldots, z_n)$ and $i,j \in \qty{1,\ldots,n}$ with $i \leq j$, define the notation $\mathbf{z}_i^j \coloneqq (z_i, \ldots, z_j)$.
    Then,
    \begin{equation}
        \dv{t} \int_{\gamma(t)} \dd[n]{\mathbf z} f(\mathbf z) = \sum_{k=1}^n \int_{\gamma(t)} \dd[n-1]{\mathbf w} \qty[ f(\mathbf{w}_1^{k-1}, t^-, \mathbf{w}_k^{n-1}) - f(\mathbf{w}_1^{k-1}, t^+, \mathbf{w}_k^{n-1}) ].
    \end{equation}
\end{theorem}
\begin{proof}
    We proceed by induction on $n$.
    In case $n=1$,
    \begin{equation}
        \dv{t} \int_{\gamma(t)} \dd{z} f(z) = \dv{t} \int_0^t \dd{\tau} \qty[ f(\tau^-) - f(\tau^+) ] = f(t^-) - f(t^+),
    \end{equation}
    where we used the fact that the vertical track, if present, does not depend on $t$ and can be omitted.
    Now assume the statement to be true for $n-1$.
    We can apply the Leibniz integral rule to write
    \begin{equation} \label{SM:eq:leibniz}
        \begin{split}
            \dv{t} \int_{\gamma(t)} \dd[n]{\mathbf z} f(\mathbf z)
            &= \dv{t} \int_{\gamma(t)} \dd{z_n} \int_{\gamma(t)} \dd[n-1]{\mathbf w} f(\mathbf{w}, z_n) \\
            &= \int_{\gamma(t)} \dd[n-1]{\mathbf w} \qty[f(\mathbf{w}, t^-) - f(\mathbf{w}, t^+)] + \int_{\gamma(t)} \dd{z_n} \dv{t} \int_{\gamma(t)} \dd[n-1]{\mathbf w} f(\mathbf{w}, z_n).
        \end{split}
    \end{equation}
    By the induction hypothesis, we have
    \begin{equation}
        \dv{t} \int_{\gamma(t)} \dd[n-1]{\mathbf w} f(\mathbf{w}, z_n)
        = \sum_{k=1}^{n-1} \int_{\gamma(t)} \dd[n-2]{\mathbf y} \qty[f(\mathbf{y}_1^{k-1}, t^-, \mathbf{y}_k^{n-2}, z_n) - f(\mathbf{y}_1^{k-1}, t^+, \mathbf{y}_k^{n-2}, z_n)],
    \end{equation}
    and therefore, if we redefine $(\mathbf{y}, z_n) \mapsto \mathbf{w}$,
    \begin{equation}
        \begin{split}
            \int_{\gamma(t)} \dd{z_n} \dv{t} \int_{\gamma(t)} \dd[n-1]{\mathbf w} f(\mathbf{w},z_n)
            = \sum_{k=1}^{n-1} \int_{\gamma(t)} \dd[n-1]{\mathbf w} \qty[f(\mathbf{w}_1^{k-1}, t^-, \mathbf{w}_k^{n-1}) - f(\mathbf{w}_1^{k-1}, t^+, \mathbf{w}_k^{n-1})].
        \end{split}
    \end{equation}
    The result follows by substituting into Eq.~\eqref{SM:eq:leibniz}.
\end{proof}

\begin{corollary}
    Let $f(\mathbf z)$ be an operator-valued function defined on $[\gamma(t)]^n$, with $n \in \mathbb{N}_+$.
    In case
    \begin{equation} \label{SM:eq:condition}
        \int_{\gamma(t)} \dd[n-1]{\mathbf w} f(\mathbf{w}_1^{k-1}, t^\pm, \mathbf{w}_k^{n-1}) = \int_{\gamma(t)} \dd[n-1]{\mathbf w} f(t^\pm, \mathbf{w})
    \end{equation}
    holds for every $k \in \qty{1,\ldots,n}$, then
    \begin{equation} \label{SM:eq:thmcalculus}
        \dv{t} \int_{\gamma(t)} \dd[n]{\mathbf z} f(\mathbf z) = n \int_{\gamma(t)} \dd[n-1]{\mathbf w} \qty[ f(t^-, \mathbf{w}) - f(t^+, \mathbf{w}) ].
    \end{equation}
\end{corollary}

Eq.~\eqref{SM:eq:thmcalculus} is what was stated in Eq.~(14) of the Letter.
Looking at Eq.~\eqref{SM:eq:rho_after_wick}, in our case the function to be considered is
\begin{equation}
    f(\mathbf z) \equiv f(z_1, \ldots, z_{2m}) = \sum_{\bm\alpha} \mathcal{C}_{1,2} \ldots \mathcal{C}_{2m-1,2m} \mathbb{T}_\zeta \qty{A_1 \ldots A_{2m} R_0},
\end{equation}
where $\sum_{\bm\alpha}$ is a sum over the interaction indices.
We now prove that this function satisfies the condition in Eq.~\eqref{SM:eq:condition}.

To do that, it is useful to introduce the ``underline'' notation mentioned in the Letter.
Specifically, let us write
\begin{equation}
    A_{\underline{i}}^\pm \equiv A_{\alpha_i}(t^\pm),
\end{equation}
so that an underlined index indicates an evaluation in $t^\pm$.
Even though the two cases with $t^\pm$ lead to the same real-time operator, it is necessary to distinguish them for ordering purposes.
Similarly, we define
\begin{equation}
    \mathcal{C}_{\underline{i},j} \equiv \mathcal{C}_{\alpha_i,\alpha_j}(t^\pm, z_j),
    \qquad
    \mathcal{C}_{i,\underline{j}} \equiv \mathcal{C}_{\alpha_i,\alpha_j}(z_i, t^\pm).
\end{equation}
In this case, we do not need to distinguish the two cases $t^\pm$ since the the correlation function assumes the same value in the two scenarios, as one can easily verify from its definition~\eqref{SM:eq:C}.

Now, fix $k \in \qty{1,\ldots,2m}$ and let $\mathbf{w} = (w_1, \ldots, w_{k-1}, w_{k+1}, \ldots, w_{2m})$.
In case $k$ is odd, the left-hand side of Eq.~\eqref{SM:eq:condition} writes
\begin{equation}
    \Sint_{\gamma(t)} \dd[2m-1]{\mathbf w} \mathcal{C}_{1,2} \ldots \mathcal{C}_{k-2,k-1} \mathcal{C}_{\underline{k},k+1} \mathcal{C}_{k+2,k+3} \ldots \mathcal{C}_{2m-1,2m} \mathbb{T}_\zeta \qty{ A_1 \ldots A_{k-1} A_{\underline{k}}^\pm A_{k+1} A_{k+2} \ldots A_{2m} R_0 }.
\end{equation}
With $k-1$ transpositions we can bring $A_{\underline{k}}^\pm$ at the beginning of the string of ordered operators, and with subsequent $k-1$ transpositions we can bring $A_{k+1}$ immediately to the right of $A_{\underline{k}}^\pm$.
Since $2(k-1)$ is even, no additional sign factor appears, and we have
\begin{equation}
    \Sint_{\gamma(t)} \dd[2m-1]{\mathbf w} \mathcal{C}_{\underline{k},k+1} \mathcal{C}_{1,2} \ldots \mathcal{C}_{k-2,k-1} \mathcal{C}_{k+2,k+3} \ldots \mathcal{C}_{2m-1,2m} \mathbb{T}_\zeta \qty{ A_{\underline{k}}^\pm A_{k+1} A_1 \ldots A_{k-1} A_{k+2} \ldots A_{2m} R_0 },
\end{equation}
which is equivalent to the right-hand side of Eq.~\eqref{SM:eq:condition} after an appropriate change of variables.
A similar argument can be carried out in case $k$ is even.
This time, the left-hand side of Eq.~\eqref{SM:eq:condition} writes
\begin{equation}
    \Sint_{\gamma(t)} \dd[2m-1]{\mathbf w} \mathcal{C}_{1,2} \ldots \mathcal{C}_{k-3,k-2} \mathcal{C}_{k-1,\underline{k}} \mathcal{C}_{k+1,k+2} \ldots \mathcal{C}_{2m-1,2m} \mathbb{T}_\zeta \qty{ A_1 \ldots A_{k-2} A_{k-1} A_{\underline{k}}^\pm A_{k+1} \ldots A_{2m} R_0 }.
\end{equation}
With $k-1$ transpositions we can again bring $A_{\underline{k}}^\pm$ at the beginning of the string, and then with $k-2$ transpositions we can bring $A_{k-1}$ immediately to the right of $A_{\underline{k}}^\pm$.
Since in total we performed $2k-3$ transpositions, which is odd, an additional factor $\zeta$ appears.
However, we can also write $\mathcal{C}_{k-1,\underline{k}} = \zeta \mathcal{C}_{\underline{k},k-1}$, which absorbs the previous $\zeta$.
The result is
\begin{equation}
    \Sint_{\gamma(t)} \dd[2m-1]{\mathbf w} \mathcal{C}_{\underline{k},k-1} \mathcal{C}_{1,2} \ldots \mathcal{C}_{k-3,k-2} \mathcal{C}_{k+1,k+2} \ldots \mathcal{C}_{2m-1,2m} \mathbb{T}_\zeta \qty{ A_{\underline{k}}^\pm A_{k-1} A_1 \ldots A_{k-2} A_{k+1} \ldots A_{2m} R_0 },
\end{equation}
which is equivalent to the right-hand side of Eq.~\eqref{SM:eq:condition} after an appropriate change of variables.

We conclude that Eq.~\eqref{SM:eq:thmcalculus} can be used to evaluate the time derivative of $M_m(t)$ in Eq.~\eqref{SM:eq:rho_after_wick}.
Note that
\begin{equation}
    \mathbb{T}_\zeta \qty{ A_{\underline{1}}^- A_2 \ldots A_{2m} R_0 } = A_{\underline{1}} \mathbb{T}_\zeta \qty{ A_2 \ldots A_{2m} R_0 },
    \qquad
    \mathbb{T}_\zeta \qty{ A_{\underline{1}}^+ A_2 \ldots A_{2m} R_0 } = \mathbb{T}_\zeta \qty{ A_2 \ldots A_{2m} R_0 } A_{\underline{1}},
\end{equation}
where we dropped the $^\pm$ notation once we are outside the ordering operation.
Therefore,
\begin{equation} \label{SM:eq:rho_after_derivative}
    \dv{M_m(t)}{t} = 2m \Sint_{\gamma(t)} \dd[2m-1]{\mathbf z} \mathcal{C}_{\underline{1},2} \mathcal{C}_{3,4} \ldots \mathcal{C}_{2m-1,2m} \qty[ A_{\underline{1}}, \mathbb{T}_\zeta \qty{A_2 \ldots A_{2m} R_0} ],
\end{equation}
which is Eq.~(15) of the Letter.


\section{Reduction of ordered products of system operators} \label{SM:sec:reduction}

At this point we introduce the quadratic system assumption: this will allow us to expand the right-hand side of Eq.~\eqref{SM:eq:rho_after_derivative} and make $M_k(t), k < m$ appear, which is what we require to close the differential equation.
The main ingredient is contained in Eq.~(16) of the Letter, which is a modified version of a statement that can be found in Ref.~\cite{Ferialdi2021Wick}.
In this section we provide a detailed proof.

As a preliminary step, we introduce a $\zeta$-graded derivative symbol $\partial_j$, which acts as follows on a product of system operators:
\begin{equation}
    \partial_{\mu_j} \qty( X_{\mu_1} \ldots X_{\mu_j} \ldots X_{\mu_n} ) \coloneqq \zeta^{j-1} X_{\mu_1} \ldots X_{\mu_{j-1}} X_{\mu_{j+1}} \ldots X_{\mu_n}.
\end{equation}
This can be interpreted as follows: we imagine to bring $X_{\mu_j}$ to the beginning of the string of operators introducing a $\zeta$ factor for each transposition; then we eliminate $X_{\mu_j}$ by acting with the derivative on it.
The following lemma clarifies that this derivative commutes with the ordering operation.
\begin{lemma} \label{SM:lemma:partial}
    For any string of system operators $X_1, \ldots, X_n$ and any $j \in \qty{1,\ldots,n}$, one has
    \begin{equation}
        \partial_j \mathbb{T}_\zeta \qty{X_1 \ldots X_n} = \zeta^{j-1} \mathbb{T}_\zeta \qty{X_1 \ldots X_{j-1} X_{j+1} \ldots X_n}.
    \end{equation}
\end{lemma}
\begin{proof}
    Let $\varphi \in \mathfrak{S}_n$ be the permutation that orders the string $X_1 \ldots X_n$.
    Such ordering can also be achieved in the following equivalent way: first, move $X_j$ at the beginning with $j-1$ transpositions; then, order the remaining operators $X_1 \ldots X_{j-1} X_{j+1} \ldots X_n$, calling $\phi$ the permutation that does that; finally, bring $X_j$ back into position with $\varphi^{-1}(j)-1$ transpositions.
    This remark allows us to say that
    \begin{equation}
        \zeta^{N(\varphi)} = \zeta^{j-1} \zeta^{N(\phi)} \zeta^{p-1}, \qquad p \coloneqq \varphi^{-1}(j).
    \end{equation}
    Since $\mathbb{T}_\zeta \qty{X_1 \ldots X_n} = \zeta^{N(\varphi)} X_{\varphi(1)} \ldots X_{\varphi(n)}$, by applying $\partial_j$ we remove $X_j$, appearing at position $p$, and we introduce an additional factor $\zeta^{p-1}$.
    Using the previous relation, and the fact that the remaining operators are already ordered,
    \begin{equation}
        \begin{split}
            \partial_j \mathbb{T}_\zeta\qty{X_1 \ldots X_n} &= \zeta^{N(\varphi)} \zeta^{p-1} X_{\varphi(1)} \ldots X_{\varphi(p-1)} X_{\varphi(p+1)} \ldots X_{\varphi(n)} \\
            &= \zeta^{N(\phi)} \zeta^{j-1} X_{\phi(1)} \ldots X_{\phi(n)} = \zeta^{j-1} \mathbb{T}_\zeta \qty{X_1 \ldots X_{j-1} X_{j+1} \ldots X_n},
        \end{split}
    \end{equation}
    where in the second equality it is understood that $X_j$ does not appear in the string $X_{\phi(1)} \ldots X_{\phi(n)}$.
\end{proof}

Now we are ready to prove the main result of this section.
\begin{lemma} \label{SM:lemma:reduction}
    Suppose $[A_i, A_j]_\zeta$ is a $c$-number for every choice of $i$ and $j$.
    Then,
    \begin{equation}
        \mathbb{T}_\zeta \qty{A_0 A_1 \ldots A_{2m} R_0} = \hat{A}_0 \mathbb{T}_\zeta \qty{A_1 \ldots A_{2m} R_0} - \sum_{j=1}^{2m} \zeta^{j-1} \Sigma_{0,j} \mathbb{T}_\zeta \qty{A_1 \ldots A_{j-1} A_{j+1} \ldots A_{2m} R_0},
    \end{equation}
    where
    \begin{equation}
        \hat{A}_0 X \coloneqq \begin{cases}
            A_0 X & z_0 \in \gamma_-(t), \\
            \zeta X A_0 & z_0 \in \gamma_+(t),
        \end{cases}
        \qquad
        \Sigma_{i,j} \coloneqq \begin{cases}
            \theta_{z_i \prec z_j} [A_i, A_j]_\zeta & z_i \in \gamma_-(t), \\
            -\theta_{z_i \succ z_j} [A_i, A_j]_\zeta & z_i \in \gamma_+(t).
        \end{cases}
    \end{equation}
\end{lemma}
\begin{proof}
    By definition of contour ordering, we can always find a permutation $\varphi \in \mathfrak{S}_{2m}$ and an index $k \in \qty{0,\ldots,2m}$ such that $z_{\varphi(1)} \succ \ldots \succ z_{\varphi(2m)}$ and $\mathbb{T}_\zeta \qty{A_1 \ldots A_{2m} R_0} = \zeta^{N(\varphi)} \zeta^{2m-k} A_{\varphi(1)} \ldots A_{\varphi(k)} R_0 A_{\varphi(k+1)} \ldots A_{\varphi(2m)}$.
    Let us distinguish the two cases $z_0 \in \gamma_\pm(t)$.

    In case $z_0 \in \gamma_-(t)$, we can find an index $q \in \qty{0,\ldots,k}$ such that
    \begin{equation}
        \mathbb{T}_\zeta\qty{A_0 A_1 \ldots A_{2m} R_0} = \zeta^{N(\varphi)} \zeta^{2m-k+q} A_{\varphi(1)} \ldots A_{\varphi(q)} A_0 A_{\varphi(q+1)} \ldots A_{\varphi(k)} R_0 A_{\varphi(k+1)} \ldots A_{\varphi(2m)}.
    \end{equation}
    Now we use the (anti)commutation relation $A_i A_0 = [A_i, A_0]_\zeta + \zeta A_0 A_i$ to bring $A_0$ at the beginning of the string of operators, using the fact that $[A_i, A_0]_\zeta$ is a $c$-number.
    Specifically,
    \begin{equation}
        \begin{split}
            A_{\varphi(1)} \ldots A_{\varphi(q)} A_0 &= [A_{\varphi(q)}, A_0]_\zeta A_{\varphi(1)} \ldots A_{\varphi(q-1)} + \zeta A_{\varphi(1)} \ldots A_{\varphi(q-1)} A_0 A_{\varphi(q)} \\
            &= \sum_{i=0}^{q-1} \zeta^i [A_{\varphi(q-i)}, A_0]_\zeta A_{\varphi(1)} \ldots A_{\varphi(q-i-1)} A_{\varphi(q-i+1)} \ldots A_{\varphi(q)} + \zeta^q A_0 A_{\varphi(1)} \ldots A_{\varphi(q)},
        \end{split}
    \end{equation}
    and therefore
    \begin{equation}
        \begin{split}
            \mathbb{T}_\zeta &\qty{A_0 A_1 \ldots A_{2m} R_0} = \zeta^{N(\varphi)} \zeta^{2m-k} A_0 A_{\varphi(1)} \ldots A_{\varphi(k)} R_0 A_{\varphi(k+1)} \ldots A_{\varphi(2m)} \\
            &\qquad\qquad - \zeta^{N(\varphi)} \zeta^{2m-k} \sum_{i=0}^{q-1} \zeta^{q-i-1} [A_0, A_{\varphi(q-i)}]_\zeta A_{\varphi(1)} \ldots A_{\varphi(q-i-1)} A_{\varphi(q-i+1)} \ldots A_{\varphi(k)} R_0 A_{\varphi(k+1)} \ldots A_{\varphi(2m)} \\
            &= A_0 \mathbb{T}_\zeta \qty{A_1 \ldots A_{2m} R_0} - \zeta^{N(\varphi)} \zeta^{2m-k} \sum_{i=0}^{q-1} [A_0, A_{\varphi(q-i)}]_\zeta \partial_{\varphi(q-i)} \qty[ A_{\varphi(1)} \ldots A_{\varphi(k)} R_0 A_{\varphi(k+1)} \ldots A_{\varphi(2m)} ] \\
            &= A_0 \mathbb{T}_\zeta \qty{A_1 \ldots A_{2m} R_0} - \sum_{i=0}^{q-1} [A_0, A_{\varphi(q-i)}]_\zeta \partial_{\varphi(q-i)} \mathbb{T}_\zeta \qty{A_1 \ldots A_{2m} R_0}.
        \end{split}
    \end{equation}
    The sum can be extended to all indices, provided we remember to keep a nonzero contribution only from those indices that are located later than $z_0$ on the contour:
    \begin{equation}
        \begin{split}
            \mathbb{T}_\zeta \qty{A_0 A_1 \ldots A_{2m} R_0} &= A_0 \mathbb{T}_\zeta \qty{A_1 \ldots A_{2m} R_0} - \sum_{j=1}^{2m} \theta_{z_0 \prec z_j} [A_0, A_j]_\zeta \partial_j \mathbb{T}_\zeta \qty{A_1 \ldots A_{2m} R_0} \\
            &= A_0 \mathbb{T}_\zeta \qty{A_1 \ldots A_{2m} R_0} - \sum_{j=1}^{2m} \zeta^{j-1} \theta_{z_0 \prec z_j} [A_0, A_j]_\zeta \mathbb{T}_\zeta \qty{A_1 \ldots A_{j-1} A_{j+1} \ldots A_{2m} R_0},
        \end{split}
    \end{equation}
    where in the second equality we employed Lemma~\ref{SM:lemma:partial}.

    A similar calculation can be performed in case $z_0 \in \gamma_+(t)$.
    This time, we can find $q \in \qty{k,\ldots,2m}$ such that
    \begin{equation}
        \mathbb{T}_\zeta \qty{A_0 A_1 \ldots A_{2m} R_0} = \zeta^{N(\varphi)} \zeta^{2m-k+q+1} A_{\varphi(1)} \ldots A_{\varphi(k)} R_0 A_{\varphi(k+1)} \ldots A_{\varphi(q)} A_0 A_{\varphi(q+1)} \ldots A_{\varphi(2m)}.
    \end{equation}
    Now we use the relation $A_0 A_i = [A_0, A_i]_\zeta + \zeta A_i A_0$ to bring $A_0$ at the end of the string of operators:
    \begin{equation}
        \begin{split}
            A_0 A_{\varphi(q+1)} \ldots A_{\varphi(2m)} = \sum_{i=1}^{2m-q} \zeta^{i+1} & [A_0, A_{\varphi(q+i)}]_\zeta A_{\varphi(q+1)} \ldots A_{\varphi(q+i-1)} A_{\varphi(q+i+1)} \ldots A_{\varphi(2m)} \\
            & + \zeta^{2m-q} A_{\varphi(q+1)} \ldots A_{\varphi(2m)} A_0,
        \end{split}
    \end{equation}
    so that
    \begin{equation}
        \begin{split}
            \mathbb{T}_\zeta & \qty{A_0 A_1 \ldots A_{2m} R_0} = \zeta^{N(\varphi)} \zeta^{2m-k} \zeta^{2m+1} A_{\varphi(1)} \ldots A_{\varphi(k)} R_0 A_{\varphi(k+1)} \ldots A_{\varphi(2m)} A_0 \\
            & \qquad\qquad + \zeta^{N(\varphi)} \zeta^{2m-k} \sum_{i=1}^{2m-q} \zeta^{q+i} [A_0, A_{\varphi(q+i)}]_\zeta A_{\varphi(1)} \ldots A_{\varphi(k)} R_0 A_{\varphi(k+1)} \ldots A_{\varphi(q+i-1)} A_{\varphi(q+i+1)} \ldots A_{\varphi(2m)} \\
            &= \zeta \mathbb{T}_\zeta\qty{A_1 \ldots A_{2m} R_0} A_0 + \zeta^{N(\varphi)} \zeta^{2m-k} \sum_{i=1}^{2m-q} [A_0, A_{\varphi(q+i)}]_\zeta \partial_{\varphi(q+i)} \qty[ A_{\varphi(1)} \ldots A_{\varphi(k)} R_0 A_{\varphi(k+1)} \ldots A_{\varphi(2m)} ] \\
            &= \zeta \mathbb{T}_\zeta\qty{A_1 \ldots A_{2m} R_0} A_0 + \sum_{i=1}^{2m-q} [A_0, A_{\varphi(q+i)}]_\zeta \partial_{\varphi(q+i)} \mathbb{T}_\zeta \qty{A_1 \ldots A_{2m} R_0}.
        \end{split}
    \end{equation}
    Again, we can extend the sum to all indices (provided we keep nonzero contribution only from indices that are earlier than $z_0$ on the contour) and apply Lemma~\ref{SM:lemma:partial}:
    \begin{equation}
        \begin{split}
            \mathbb{T}_\zeta \qty{A_0 A_1 \ldots A_{2m} R_0} &= \zeta \mathbb{T}_\zeta \qty{A_1 \ldots A_{2m} R_0} A_0 + \sum_{j=1}^{2m} \theta_{z_0 \succ z_j} [A_0, A_j]_\zeta \partial_j \mathbb{T}_\zeta \qty{A_1 \ldots A_{2m} R_0} \\
            &= \zeta \mathbb{T}_\zeta \qty{A_1 \ldots A_{2m} R_0} A_0 + \sum_{j=1}^{2m} \zeta^{j-1} \theta_{z_0 \succ z_j} [A_0, A_j]_\zeta \mathbb{T}_\zeta \qty{A_1 \ldots A_{j-1} A_{j+1} \ldots A_{2m} R_0}.
        \end{split}
    \end{equation}
\end{proof}


\section{Closing the differential equation} \label{SM:sec:closing}

In this section we show that if we apply the reduction lemma~\ref{SM:lemma:reduction} to Eq.~\eqref{SM:eq:rho_after_derivative} we close the differential equation and we obtain the exact master equation in Eq.~(20) of the Letter.

With a direct substitution, we first write
\begin{equation}
    \begin{split}
        \dv{M_m(t)}{t} &= 2m \Sint_{\gamma(t)} \dd[2m-1]{\mathbf z} \mathcal{C}_{\underline{1},2} \mathcal{C}_{3,4} \ldots \mathcal{C}_{2m-1,2m} \qty[A_{\underline{1}}, \hat{A}_2 \mathbb{T}_\zeta\qty{A_3 \ldots A_{2m} R_0}] \\
        &- 2m \sum_{j=3}^{2m} \zeta^{j-1} \Sint_{\gamma(t)} \dd[2m-1]{\mathbf z} \mathcal{C}_{\underline{1},2} \mathcal{C}_{3,4} \ldots \mathcal{C}_{2m-1,2m} \Sigma_{2,j} \qty[A_{\underline{1}}, \mathbb{T}_\zeta\qty{A_3 \ldots A_{j-1} A_{j+1} \ldots A_{2m} R_0}].
    \end{split}
\end{equation}
In the first term, we recognize the appearance of $M_{m-1}(t)$ if we isolate the variables $z_3, \ldots, z_{2m}$ [cfr.~\eqref{SM:eq:rho_after_wick}].
For what concerns the second term, we can actually show that the sum over $j$ is trivial to perform.
Specifically, we now prove that all terms of such sum are equal to each other.
When $j$ is even, the term is
\begin{equation}
    \begin{split}
        & \zeta^{j-1} \Sint_{\gamma(t)} \dd[2m-1]{\mathbf z} \mathcal{C}_{\underline{1},2} \mathcal{C}_{3,4} \ldots \mathcal{C}_{j-1,j} \ldots \mathcal{C}_{2m-1,2m} \Sigma_{2,j} \qty[A_{\underline{1}}, \mathbb{T}_\zeta\qty{A_3 A_4 \ldots A_{j-1} A_{j+1} \ldots A_{2m} R_0}] \\
        &= \zeta \Sint_{\gamma(t)} \dd[2m-1]{\mathbf z} \mathcal{C}_{\underline{1},2} \mathcal{C}_{j,j-1} \ldots \mathcal{C}_{4,3} \ldots \mathcal{C}_{2m-1,2m} \Sigma_{2,3} \qty[A_{\underline{1}}, \mathbb{T}_\zeta\qty{A_j A_{j-1} A_5 \ldots A_{j-2} A_4 A_{j+1} \ldots A_{2m} R_0}],
    \end{split}
\end{equation}
where we performed the change of variables $j \leftrightarrow 3$ and $(j-1) \leftrightarrow 4$.
Inside the contour ordering, we can now exchange $A_4$ with $A_{j-1}$, and then we can move $A_j$ to the immediate right of $A_{j-1}$ using $j-4$ transpositions.
In total, these are $j-5$ transpositions, which is odd.
Therefore, using also the fact that $\mathcal{C}_{j,j-1} \mathcal{C}_{4,3} = \mathcal{C}_{j-1,j} \mathcal{C}_{3,4}$, we end up with
\begin{equation}
    \Sint_{\gamma(t)} \dd[2m-1]{\mathbf z} \mathcal{C}_{\underline{1},2} \mathcal{C}_{3,4} \ldots \mathcal{C}_{j-1,j} \ldots \mathcal{C}_{2m-1,2m} \Sigma_{2,3} \qty[A_{\underline{1}}, \mathbb{T}_\zeta\qty{A_4 A_5 \ldots A_{j-1} A_j A_{j+1} \ldots A_{2m} R_0}],
\end{equation}
which can be expressed without explicit mention to $j$.
A similar manipulation can be performed in case $j$ is odd, where the term to consider is
\begin{equation}
    \begin{split}
        & \zeta^{j-1} \Sint_{\gamma(t)} \dd[2m-1]{\mathbf z} \mathcal{C}_{\underline{1},2} \mathcal{C}_{3,4} \ldots \mathcal{C}_{j,j+1} \ldots \mathcal{C}_{2m-1,2m} \Sigma_{2,j} \qty[A_{\underline{1}}, \mathbb{T}_\zeta\qty{A_3 A_4 \ldots A_{j-1} A_{j+1} \ldots A_{2m} R_0}] \\
        &= \Sint_{\gamma(t)} \dd[2m-1]{\mathbf z} \mathcal{C}_{\underline{1},2} \mathcal{C}_{j,j+1} \ldots \mathcal{C}_{3,4} \ldots \mathcal{C}_{2m-1,2m} \Sigma_{2,3} \qty[A_{\underline{1}}, \mathbb{T}_\zeta\qty{ A_j A_{j+1} A_5 \ldots A_{j-1} A_4 A_{j+2} \ldots A_{2m} R_0 }],
    \end{split}
\end{equation}
where we made the change of variables $j \leftrightarrow 3$ and $(j+1) \leftrightarrow 4$.
As before, with $j-5$ transpositions we can put the system operators back in place (a transposition to exchange $A_4$ and $A_{j+1}$, and $j-4$ transpositions to move $A_j$ to the immediate left of $A_{j+1}$), but this time $j-5$ is even and no additional $\zeta$ factor is generated in the process.

As a consequence of these arguments,
\begin{equation}
    \begin{split}
        \dv{M_m(t)}{t} &= 2m \Sint_{\gamma(t)} \dd{z_2} \mathcal{C}_{\underline{1},2} \qty[A_{\underline{1}}, \hat{A}_2 M_{m-1}(t)] \\
        &- 2m(2m-2) \Sint_{\gamma(t)} \dd[2m-1]{\mathbf z} \mathcal{C}_{\underline{1},2} \Sigma_{2,3} \mathcal{C}_{3,4} \ldots \mathcal{C}_{2m-1,2m} \qty[A_{\underline{1}}, \mathbb{T}_\zeta\qty{A_4 \ldots A_{2m} R_0}],
    \end{split}
\end{equation}
Lemma~\ref{SM:lemma:reduction} can now be applied again to reduce $\mathbb{T}_\zeta\qty{A_4 \ldots A_{2m} R_0}$, generating a term containing $M_{m-2}(t)$ and a term containing $\mathbb{T}_\zeta\qty{A_6 \ldots A_{2m} R_0}$.
The pattern repeats until we exhaust all operators inside the contour ordering.
The result is the one reported in Eqs.~(18)-(19) of the Letter, repeated here for convenience:
\begin{equation}
    \dv{M_m(t)}{t} = \sum_{k=1}^m \frac{(-1)^{k+1}(2m)!!}{(2m-2k)!!} \hat{\mathcal S}_k(t) M_{m-k}(t),
\end{equation}
where
\begin{equation} \label{SM:eq:S}
    \hat{\mathcal S}_k(t)X \coloneqq \Sint_{\gamma(t)} \dd[2k-1]{\mathbf z} \mathcal{C}_{\underline{1},2} \Sigma_{2,3} \mathcal{C}_{3,4} \Sigma_{4,5} \ldots \Sigma_{2k-2,2k-1} \mathcal{C}_{2k-1,2k} \qty[A_{\underline{1}}, \hat{A}_{2k} X].
\end{equation}

If we remember the formula $(2x)!! = x! 2^x$ and Eq.~\eqref{SM:eq:rho_after_wick},
\begin{equation}
    \dv{\varrho(t)}{t} = \sum_{m=1}^\infty \frac{(-1)^m}{m! 2^m} \dv{M_m(t)}{t} = \sum_{m=1}^\infty \sum_{k=1}^m \frac{(-1)^{m+k+1}}{(m-k)!2^{m-k}} \hat{\mathcal S}_k(t) M_{m-k}(t).
\end{equation}
We recognize here the Cauchy product of two series, in the form
\begin{equation}
    \sum_{m=1}^\infty \sum_{k=1}^m F(m,k) = \sum_{k=1}^\infty \sum_{m=k}^\infty F(m,k) = \sum_{k=1}^\infty \sum_{m=0}^\infty F(m+k,k),
\end{equation}
therefore
\begin{equation} \label{SM:eq:exactME}
    \dv{\varrho(t)}{t} = \sum_{k=1}^\infty \sum_{m=0}^\infty \frac{(-1)^{m+1}}{m!2^m} \hat{\mathcal S}_k(t) M_m(t) = -\sum_{k=1}^\infty \hat{\mathcal S}_k(t) \varrho(t),
\end{equation}
which is the exact master equation reported in Eq.~(20) of the Letter.


\section{From contour time to physical time} \label{SM:sec:physical}

If we isolate the operator dependence in Eq.~\eqref{SM:eq:S}, the exact master equation~\eqref{SM:eq:exactME} can be written as
\begin{equation}
    \dv{\varrho(t)}{t} = -\sum_{\alpha,\beta} \int_{\gamma(t)} \dd{z} \mathcal{G}_{\alpha\beta}(t^+,z) [A_\alpha(t), \hat{A}_\beta(z) \varrho(t)],
\end{equation}
where we reintroduced the full notation for the sake of clarity and we defined the ``dressed'' environment correlation function $\mathcal{G}_{\alpha\beta}(z,w) = \sum_{k=1}^\infty \mathcal{G}_{\alpha\beta}^{(k)}(z,w)$ by the recursion
\begin{equation} \label{SM:eq:recursion}
    \mathcal{G}^{(1)}_{\alpha\beta}(z,w) = \mathcal{C}_{\alpha\beta}(z,w),
    \qquad
    \mathcal{G}^{(k)}_{\alpha\beta}(z,w) = \sum_{\mu_1,\mu_2} \int_{\gamma(t)} \dd[2]{\mathbf y} \mathcal{G}^{(k-1)}_{\alpha\mu_1}(z,y_1) \Sigma_{\mu_1\mu_2}(y_1,y_2) \mathcal{C}_{\mu_2\beta}(y_2,w),
\end{equation}
which directly follows from the definition of $\hat{\mathcal S}_k(t)$ in Eq.~\eqref{SM:eq:S}.

We now want to show how to recover Eq.~(23) of the Letter by moving back to physical time.
If we decompose $\gamma(t)$ in its branches $\gamma_\pm(t)$,
\begin{equation} \label{SM:eq:physicalME}
    \dv{\varrho(t)}{t} = - \sum_{\alpha,\beta} \int_{\gamma_-(t)} \dd{z} \mathcal{G}_{\alpha\beta}(t^+, z) [A_\alpha(t), A_\beta(z) \varrho(t)] - \zeta \sum_{\alpha,\beta} \int_{\gamma_+(t)} \dd{z} \mathcal{G}_{\alpha\beta}(t^+, z) [A_\alpha(t), \varrho(t) A_\beta(z)].
\end{equation}
which is almost the same as Eq.~(23) in the Letter.
To conclude we only need to show that the second term in Eq.~\eqref{SM:eq:physicalME} is equal to the Hermitian conjugate of the first one.
To do that, we need some preliminary work.

Since the system-environment interaction $V = \sum_\alpha A_\alpha \otimes B_\alpha$ is a Hermitian operator, for every index $\alpha$ we can find a unique index $\overline{\alpha}$ such that $A^\dagger_\alpha = A_{\overline\alpha}$ and $B^\dagger_\alpha = B_{\overline\alpha}$.
Using the definition of interaction picture, one then sees that
\begin{equation}
    A^\dagger_\alpha(z) = A_{\overline\alpha}(z^*),
    \qquad
    B^\dagger_\alpha(z) = B_{\overline\alpha}(z^*),
\end{equation}
where $z^*$ stands for $z$ but flipped on the opposite branch of the contour.
This observation allows us to formulate the following symmetry property of the dressed correlation function.
\begin{lemma} \label{SM:lemma:G_symmetry}
    For every indices $\alpha,\beta$ and for every $z,w \in \gamma(t)$,
    \begin{equation} \label{SM:eq:symmetry}
        \mathcal{G}^*_{\alpha,\beta}(z,w) = \zeta \mathcal{G}_{\overline{\alpha}\overline{\beta}}(z^*, w^*).
    \end{equation}
\end{lemma}
\begin{proof}
    It is sufficient to prove that the property~\eqref{SM:eq:symmetry} is satisfied by all terms $\mathcal{G}^{(k)}_{\alpha\beta}(z,w)$ of the series, and we can do that by induction on $k$.
    In case $k=1$, we can write
    \begin{equation}
        \begin{split}
            \mathcal{C}^*_{\alpha\beta}(z,w) &= \Tr[\mathbb{T}_\zeta\qty{B_\alpha(z) B_\beta(w) \Omega_0}]^* = \Tr[\mathbb{T}_\zeta\qty{B_{\overline\beta}(w^*) B_{\overline\alpha}(z^*) \Omega_0}] \\
            &= \mathcal{C}_{\overline{\beta}\overline{\alpha}}(w^*,z^*) = \zeta \mathcal{C}_{\overline{\alpha}\overline{\beta}}(z^*,w^*).
        \end{split}
    \end{equation}
    A similar property can be proved for the self-energy.
    In fact, assuming $z \in \gamma_-(t)$, we have
    \begin{equation}
        \begin{split}
            \Sigma^*_{\alpha\beta}(z,w) &= \theta_{z \prec w} [A_\alpha(z), A_\beta(w)]^*_\zeta = \theta_{z \prec w} [A_\beta^\dagger(w), A_\alpha^\dagger(z)]_\zeta = \theta_{z \prec w} [A_{\overline\beta}(w^*), A_{\overline\alpha}(z^*)]_\zeta \\
            &= -\zeta \theta_{z \prec w} [A_{\overline\alpha}(z^*), A_{\overline\beta}(w^*)]_\zeta = -\zeta \theta_{z^* \succ w^*} [A_{\overline\alpha}(z^*), A_{\overline\beta}(w^*)]_\zeta = \zeta \Sigma_{\overline{\alpha}\overline{\beta}}(z^*,w^*),
        \end{split}
    \end{equation}
    and a similar calculation can be performed in case $z \in \gamma_+(t)$.
    Therefore, if we assume to property to hold for $k-1$ we have
    \begin{equation}
        \begin{split}
            \qty[ \mathcal{G}^{(k)}_{\alpha\beta}(z,w) ]^* &= \zeta \sum_{\mu_1,\mu_2} \int_{\gamma(t)} \dd[2]{\mathbf y} \mathcal{G}^{(k-1)}_{\overline{\alpha}\overline{\mu}_1}(z^*,y_1^*) \Sigma_{\overline{\mu}_1\overline{\mu}_2}(y_1^*, y_2^*) \mathcal{C}_{\overline{\mu}_2\overline{\beta}}(y_2^*,w^*) \\
            &= \zeta \sum_{\mu_1,\mu_2} \int_{\gamma(t)} \dd[2]{\mathbf y} \mathcal{G}^{(k-1)}_{\overline{\alpha}\mu_1}(z^*,y_1) \Sigma_{\mu_1\mu_2}(y_1, y_2) \mathcal{C}_{\mu_2\overline{\beta}}(y_2,w^*) = \zeta \mathcal{G}^{(k)}_{\overline{\alpha}\overline{\beta}}(z^*,w^*),
        \end{split}
    \end{equation}
    where in the second equality we performed a change of variables $(y_1,y_2) \mapsto (y_1^*, y_2^*)$ and $(\mu_1,\mu_2) \mapsto (\overline{\mu}_1, \overline{\mu}_2)$, which is harmless since the contour $\gamma(t)$ contains $y^*$ if and only if it contains $y$, and we are summing over all interaction indices for $\mu_1$ and $\mu_2$.
\end{proof}

We also need the following immediate fact.
\begin{lemma} \label{SM:lemma:G_border}
    For every $\alpha,\beta$ and every $z \in \gamma(t)$, we have $\mathcal{G}_{\alpha\beta}(t^-, z) = \mathcal{G}_{\alpha\beta}(t^+, z)$.
\end{lemma}
\begin{proof}
    This follows immediately by induction from Eq.~\eqref{SM:eq:recursion} with the base case $\mathcal{C}_{\alpha\beta}(t^-,z) = \mathcal{C}_{\alpha\beta}(t^+,z)$, which can be easily verified from the definition~\eqref{SM:eq:C}.
\end{proof}
Using the two previous lemmas, we can now prove the initial claim:
\begin{equation}
    \begin{split}
        \qty{-\sum_{\alpha,\beta} \int_{\gamma_-(t)} \dd{z} \mathcal{G}_{\alpha,\beta}(t^+,z) [A_\alpha(t), A_\beta(z) \varrho(t)]}^\dag = -\sum_{\alpha,\beta} \int_{\gamma_-(t)} \dd{z} \mathcal{G}^*_{\alpha,\beta}(t^+, z) [\varrho(t) A_\beta^\dag(z), A_\alpha^\dag(t)] \\
        = -\zeta \sum_{\alpha,\beta} \int_{\gamma_-(t)} \dd{z} \mathcal{G}_{\overline{\alpha},\overline{\beta}}(t^-, z^*) [\varrho(t) A_{\overline{\beta}}(z^*), A_{\overline{\alpha}}(t)] = \zeta \sum_{\alpha,\beta} \int_{\gamma_+(t)} \dd{z} \mathcal{G}_{\alpha,\beta}(t^+, z) [\varrho(t) A_\beta(z), A_\alpha(t)].
    \end{split}
\end{equation}
Note that when $\gamma_-(t)$ is expanded in its horizontal and vertical components,
\begin{equation}
    \dv{\varrho(t)}{t} = -\sum_{\alpha,\beta} \int_0^t \dd{\tau} \mathcal{G}_{\alpha,\beta}(t^+, \tau^-) [A_\alpha(t), A_\beta(\tau) \varrho(t)] + i \sum_{\alpha,\beta} \int_0^b \dd{\lambda} \mathcal{G}_{\alpha,\beta}(t^+, -i\lambda) [A_\alpha(t), A_\beta(-i\lambda) \varrho(t)] + \text{H.c.}
\end{equation}
which reveals that the master equation has the Redfield shape, plus a contribution originated by initial system-environment correlations.

It is also convenient to express the Dyson equation,
\begin{equation} \label{SM:eq:dyson}
    \mathcal{G}(t^+, z) = \mathcal{C}(t^+, z) + \int_{\gamma(t)} \dd[2]{\mathbf w} \mathcal{G}(t^+, w_1) \Sigma(w_1, w_2) \mathcal{C}(w_2, z)
\end{equation}
in terms of physical times only.
Let us focus for simplicity on the case with factorized initial state, so that the vertical track in the contour can be ignored: similar manipulations are certainly possible in the more general case.
For a general contour function $f(z,w)$ we introduce the notations
\begin{equation}
    f^>(\tau, s) \coloneqq f(\tau^+, s^-),
    \quad
    f^<(\tau, s) \coloneqq f(\tau^-, s^+),
    \quad
    f^T(\tau, s) \coloneqq f(\tau^-, s^-),
    \quad
    f^{\widetilde T}(\tau, s) \coloneqq f(\tau^+, s^+),
\end{equation}
which are reminiscent of the standard notations used in many-body theory for the greater, lesser, time-ordered, and anti-time-ordered Keldysh components~\cite{Stefanucci2013Book}.
If we expand the contour in its branches, Eq.~\eqref{SM:eq:dyson} decomposes into the following pair of physical-time equations:
\begin{subequations} \label{SM:eq:dyson_real_pair}
    \begin{gather}
        \mathcal{G}^> = \mathcal{C}^> + \mathcal{G}^> \cdot \Sigma^T \cdot \mathcal{C}^T - \mathcal{G}^> \cdot \Sigma^< \cdot \mathcal{C}^> - \mathcal{G}^{\widetilde T} \cdot \Sigma^> \cdot \mathcal{C}^T + \mathcal{G}^{\widetilde T} \cdot \Sigma^{\widetilde T} \cdot \mathcal{C}^>, \\
        \mathcal{G}^{\widetilde T} = \mathcal{C}^{\widetilde T} + \mathcal{G}^> \cdot \Sigma^T \cdot \mathcal{C}^< - \mathcal{G}^> \cdot \Sigma^< \cdot \mathcal{C}^{\widetilde T} - \mathcal{G}^{\widetilde T} \cdot \Sigma^> \cdot \mathcal{C}^< + \mathcal{G}^{\widetilde T} \cdot \Sigma^{\widetilde T} \cdot \mathcal{C}^{\widetilde T},
    \end{gather}
\end{subequations}
where the argument $(t,\tau)$ is understood throughout and
\begin{equation}
    (X \cdot Y \cdot Z)(t,\tau) \coloneqq \int_0^t \dd[2]{\mathbf s} X(t,s_1) Y(s_1,s_2) Z(s_2,\tau).
\end{equation}

We can make some simplifications.
First of all, remembering the definition of $\Sigma$, we actually realize that $\Sigma^> = \Sigma^< = 0$ and that $\Sigma^{\widetilde T} = -\Sigma^T$.
Moreover, thanks to Lemmas~\ref{SM:lemma:G_symmetry} and~\ref{SM:lemma:G_border},
\begin{equation}
    \mathcal{G}^{\widetilde T}_{\alpha,\beta}(t,\tau) = \mathcal{G}_{\alpha,\beta}(t^+, \tau^+) = \mathcal{G}_{\alpha,\beta}(t^-, \tau^+) = \zeta \qty[\mathcal{G}_{\overline{\alpha},\overline{\beta}}(t^+, \tau^-)]^*
    = \zeta \qty[\mathcal{G}^>_{\overline{\alpha},\overline{\beta}}(t,\tau)]^*.
\end{equation}
If we define an ``effective conjugate'' operation as
\begin{equation}
    [F]^c_{\alpha,\beta} \coloneqq \zeta F^*_{\overline{\alpha},\overline{\beta}},
\end{equation}
we can then write Eqs.~\eqref{SM:eq:dyson_real_pair} as a single equation:
\begin{equation}
    \mathcal{G}^> = \mathcal{C}^> + \mathcal{G}^> \cdot \Sigma^T \cdot \mathcal{C}^T - [\mathcal{G}^>]^c \cdot \Sigma^T \cdot \mathcal{C}^>.
\end{equation}
Using a series representation $\mathcal{G}^> = \sum_{k=1}^\infty \mathcal{G}^>_k$, this can readily be transformed in the recursion
\begin{equation}
    \mathcal{G}^>_1 = \mathcal{C}^>,
    \qquad
    \mathcal{G}^>_k = \mathcal{G}^>_{k-1} \cdot \Sigma^T \cdot \mathcal{C}^T - [\mathcal{G}^>_{k-1}]^c \cdot \Sigma^T \cdot \mathcal{C}^>,
\end{equation}
Note that in terms of the known standard functions
\begin{equation}
    \sigma_{\alpha\beta}(\tau,s) \coloneqq [A_\alpha(\tau), A_\beta(s)]_\zeta,
    \qquad
    c_{\alpha\beta}(\tau, s) \coloneqq \Tr[B_\alpha(\tau) B_\beta(s) \Omega_0]
\end{equation}
we simply have
\begin{subequations}
    \begin{gather}
        \Sigma^T_{\alpha\beta}(\tau, s) = \theta(s-\tau) \sigma_{\alpha\beta}(\tau,s), \\
        \mathcal{C}^>_{\alpha\beta}(\tau, s) = c_{\alpha\beta}(\tau, s),
        \qquad
        \mathcal{C}^<_{\alpha\beta}(\tau, s) = \zeta c_{\beta\alpha}(s,\tau), \\
        \mathcal{C}^T_{\alpha\beta}(\tau, s) = \theta(\tau-s) \mathcal{C}^>_{\alpha\beta}(\tau, s) + \theta(s-\tau) \mathcal{C}^<_{\alpha\beta}(\tau, s).
    \end{gather}
\end{subequations}


\section{Master equation for the covariance matrix} \label{SM:sec:covariance}

Since the setting we are considering is quadratic in the system variables, we expect to be able to turn the master equation into an evolution equation for the covariance matrix associated with the system's state.
It is indeed possible to do so.
As a first step, let us revert the interaction picture:
\begin{equation}
    \begin{split}
        \dv{\rho(t)}{t} + i[H_S, \rho(t)] &= - \sum_{\alpha,\beta} \int_{\gamma_-(t)} \dd{z} \mathcal{G}_{\alpha\beta}(t^+, z) U_0(t,0) [A_\alpha(t), A_\beta(z) \varrho(t)] U_0(0,t) + \text{H.c.} \\
        &= -\sum_{\alpha,\beta} \int_{\gamma_-(t)} \dd{z} \mathcal{G}_{\alpha\beta}(t^+, z) [A_\alpha, A_\beta(z;t) \rho(t)] + \text{H.c.}
    \end{split}
\end{equation}
where $A_\beta(z; t) \coloneqq W_0(t^-,z) A_\beta W_0(z, t^-)$.
In case $z = \tau^-$ lies on the horizontal track, $A_\beta(\tau^-; t) = A_\beta(\tau-t)$.
Now, if we impose $\Tr[\dot{\varrho} O] = \Tr\,[\varrho \dot{O}]$, a simple computation reveals that the adjoint master equation for an observable $O$ is given by
\begin{equation} \label{SM:eq:adjoint}
    \dv{O(t)}{t} - i[H_S, O(t)] = -\sum_{\alpha,\beta} \int_{\gamma_-(t)} \dd{z} \mathcal{G}_{\alpha\beta}(t^+, z) [O(t), A_\alpha] A_\beta(z; t) + \text{H.c.}
\end{equation}

It is now convenient to introduce self-adjoint operators $\{w_i\}$ constructed using linear combinations of system ladder operators, such that
\begin{equation} \label{SM:eq:Omega}
    [w_i, w_j]_\zeta = \Omega_{ij} \mathbbm{1},
    \quad
    \Omega^T = \Omega^* = -\zeta\Omega,
\end{equation}
with $\Omega$ being a matrix of scalar coefficients.
In the bosonic case, we could take position and momentum operators and $\Omega$ would be proportional the standard symplectic matrix~\cite{Serafini2023Book}.
In the fermionic case, we could take Majorana operators and $\Omega$ would be the identity~\cite{Surace2022FermionGauss}.
We want to write the adjoint master equation corresponding to the covariance matrix~\cite{Barthel2022Third}
\begin{equation}
    \Gamma_{ij} \coloneqq \Tr[(w_i w_j + \zeta w_j w_i) \varrho] = 2 \Tr[w_i w_j \varrho] - \Omega_{ij},
    \qquad
    \Gamma^T = \Gamma^* = \zeta \Gamma.
\end{equation}

Recall that the operators $A_\alpha$ were assumed to be linear in the system ladder operators.
This means that
\begin{equation}
    A_\alpha = \sum_i \mathbb{A}_{\alpha i}w_i,
    \qquad
    A_\beta(z;t) = \sum_j \mathbb{A}_{\beta j}(z; t) w_j
\end{equation}
for appropriate coefficients $\mathbb{A}$ and $\mathbb{A}(z;t)$.
Moreover, since $H_S$ is quadratic in ladder operators, we can always write
\begin{equation}
    H_S = \sum_{i,j} h_{i j} w_i w_j,
    \qquad
    h^\dag = h = \zeta h^T.
\end{equation}
Therefore, Eq.~\eqref{SM:eq:adjoint} writes
\begin{equation} \label{SM:eq:adjoint_w}
    \dv{O(t)}{t} - i \sum_{i,j} h_{i j} [w_i w_j, O(t)] = \sum_{i,j} M_{i j}(t) [O(t), w_i] w_j + \text{H.c.},
\end{equation}
where
\begin{equation} \label{SM:eq:M}
    M(t) \coloneqq -\int_{\gamma_-(t)} \dd{z} \mathbb{A}^T \mathcal{G}(t^+, z) \mathbb{A}(z; t).
\end{equation}
Let us now write Eq.~\eqref{SM:eq:adjoint_w} using the observable $\hat{\Gamma}_{kq} \coloneqq 2w_k w_q - \Omega_{kq}\mathbbm{1}$, so that $\Gamma_{kq} = \Tr\,[\hat{\Gamma}_{kq} \varrho]$.
With a straightforward calculation that employs the (anti)commutation relations~\eqref{SM:eq:Omega}, one arrives at the following continuous differential Lyapunov equation:
\begin{equation} \label{SM:eq:lyapunov}
    \dv{\Gamma}{t} = X\Gamma + \Gamma X^T + Y,
\end{equation}
where
\begin{equation}
    X \coloneqq 2 \Re[\Omega M] - 2i \Omega h,
    \qquad
    Y \coloneqq -\Omega(M + \zeta M^T + M^\dag + \zeta M^*)\Omega.
\end{equation}


\section{Details on the example}\label{SM:sec:example}

In this final section we provide additional details on how the numerical computation for the example presented in the Letter has been carried out.
Let us start by considering a generic fermionic quadratic system Hamiltonian:
\begin{equation}
    H_S = \sum_{n,m} \qty[J_{nm} a_n^\dag a_m + \frac{1}{2} \Delta_{nm} a_n^\dag a_m^\dag - \frac{1}{2} \Delta^*_{nm} a_n a_m],
    \qquad
    J^\dag = J, \quad \Delta^T = -\Delta.
\end{equation}
If we define Majorana operators through
\begin{equation}
    w_{2n-1} = \frac{1}{\sqrt{2}} \qty(a_n^\dag + a_n),
    \qquad
    w_{2n} = \frac{i}{\sqrt{2}} \qty(a_n^\dag - a_n),
\end{equation}
we can also write $H_S = \sum_{i,j} h_{i j} w_i w_j$, where $h$ is a pure imaginary skew-symmetric matrix such that the block corresponding to the indices $(2n-1,2m-1)$, $(2n-1,2m)$, $(2n,2m-1)$, $(2n,2m)$ is found to be
\begin{equation}
    2i \begin{bmatrix}
       \Im J_{n m} + \Im \Delta_{n m} & \Re J_{n m} - \Re \Delta_{n m} \\
       - \Re J_{n m} - \Re\Delta_{n m} & \Im J_{n m} - \Im\Delta_{n m}
    \end{bmatrix}.
\end{equation}

For the environment, we take a single bath of fermionic oscillators $H_E = \sum_r \varepsilon_r c_r^\dag c_r$ initially kept in its ground state with negative large chemical potential, so that $\langle c_r c_r^\dag \rangle = 1$.
For the interaction, $V = \sum_{n \in \mathcal{I}} \sum_r \qty[g_r c_r a_n^\dag + g^*_r a_n c_r^\dag]$, where $\mathcal{I}$ is the set of system indices establishing which part of the system is coupled to the bath.
Using known results in multi-linear fermionic algebra~\cite{Cirio2022Influence}, this interaction can be put in the following tensor form:
\begin{equation}
    V = B \otimes A^\dag + B^\dag \otimes A,
    \qquad
    A = \sum_{n \in \mathcal{I}} a_n,
    \quad
    B = \sum_r g_r c_r \mathcal{P}_E,
\end{equation}
where $\mathcal{P}_E = (-1)^{\sum_r c_r^\dag c_r}$ is the parity operator of the environment.

With the ordering $B_1 = B$, $B_2 = B^\dag$, we immediately obtain
\begin{equation}
    c(\tau, s) = \begin{bmatrix}
        0 & 1 \\ 0 & 0
    \end{bmatrix} \sum_r |g_r|^2 e^{-i\varepsilon_r(\tau-s)} \simeq
    \begin{bmatrix}
        0 & 1 \\ 0 & 0
    \end{bmatrix} \frac{\gamma\lambda}{2} e^{-\lambda |\tau-s|}
\end{equation}
in case we assume the environment to have a Lorentzian spectral density.
For the system part, we first need to evolve $A_1 = A^\dag$ and $A_2 = A$ in interaction picture.
To do that, we use the Bogoliubov transformation that diagonalizes $H_S$:
\begin{equation}
    a_n = \sum_k \qty[\mathcal{A}_{nk} b_k + \mathcal{B}_{nk} b_k^\dag],
    \qquad
    H_S = \sum_k \omega_k b_k^\dag b_k + E_0.
\end{equation}
The quantities $\mathcal{A}$, $\mathcal{B}$, and $\omega_k$ can be found with a straightforward matrix diagonalization in Nambu space.
One can in fact prove that~\cite{Xiao2009Quadratic}
\begin{equation}
    \begin{bmatrix}
        J & \Delta \\ \Delta^\dag & -J^T
    \end{bmatrix}
    = U \begin{bmatrix}
        \Omega & 0 \\ 0 & -\Omega
    \end{bmatrix} U^\dag,
    \quad
    U = \begin{bmatrix}
        \mathcal{A} & \mathcal{B} \\
        \mathcal{B}^* & \mathcal{A}^*
    \end{bmatrix},
    \quad
    \Omega = \mathrm{diag}(\{\omega_k\}).
\end{equation}
Now, we simply have
\begin{equation}
    A(t) = \sum_k \qty[\varphi_k e^{-i\omega_k t} b_k + \chi_k e^{i\omega_k t} b_k^\dag],
    \qquad
    \varphi_k \coloneqq \sum_{n \in \mathcal{I}} \mathcal{A}_{n k},
    \quad
    \chi_k \coloneqq \sum_{n \in \mathcal{I}} \mathcal{B}_{n k}.
\end{equation}
Therefore, we find with some calculation that
\begin{equation}
    \sigma(\tau, s) = \sum_k \qty{
        \begin{bmatrix}
            \varphi^*_k \chi^*_k & |\chi_k|^2 \\
            |\varphi_k|^2 & \varphi_k \chi_k
        \end{bmatrix} e^{-i\omega_k(\tau-s)}
        + \begin{bmatrix}
            \varphi^*_k \chi^*_k & |\varphi_k|^2 \\
            |\chi_k|^2 & \varphi_k \chi_k
        \end{bmatrix} e^{i\omega_k(\tau-s)}
    }.
\end{equation}
With this information, the Dyson equation can be solved (assuming the initial state to be factorized, for simplicity).
From a practical standpoint, we fixed a discrete time grid and we approximated the integrals using a double composite trapezoidal rule.

Once $\mathcal{G}$ is found, the Lyapunov equation for the covariance matrix can be approached.
First we need to find the coefficient matrices $\mathbb{A}$ and $\mathbb{A}(z;t)$.
Since in the present case the vertical track is missing, we can just focus on a coefficient $\mathbb{A}(\tau)$ calculated from the standard interaction-picture operator $A_\alpha(\tau)$.
Then, we would simply have $\mathbb{A} = \mathbb{A}(0)$ and $\mathbb{A}(\tau^-; t) = \mathbb{A}(\tau - t)$.
One uses the following basis transformation between Bogoliubov and Majorana operators:
\begin{equation}
    \begin{bmatrix}
        b_k \\ b_k^\dag
    \end{bmatrix}
    = \frac{1}{\sqrt{2}}
    \sum_n \begin{bmatrix}
        \Phi^*_{n k} & i \Psi^*_{n k} \\
        \Phi_{n k} & -i \Psi_{n k}
    \end{bmatrix}
    \begin{bmatrix}
        w_{2n-1} \\ w_{2n}
    \end{bmatrix},
\end{equation}
where $\Phi \coloneqq \mathcal{A} + \mathcal{B}^*$ and $\Psi \coloneqq \mathcal{A} - \mathcal{B}^*$~\cite{DAbbruzzo2021Kitaev}.
Substituting above, one concludes that
\begin{gather}
    \mathbb{A}_{1,2n-1}(t) = \frac{1}{\sqrt 2} \sum_k \qty[\Phi^*_{n k} \chi^*_k e^{-i\omega_k t} + \Phi_{n k} \varphi^*_k e^{i\omega_k t}],
    \qquad
    \mathbb{A}_{2,2n-1}(t) = \mathbb{A}^*_{1,2n-1}(t), \\
    \mathbb{A}_{1,2n}(t) = \frac{i}{\sqrt 2} \sum_k \qty[\Psi^*_{n k} \chi^*_k e^{-i\omega_k t} - \Psi_{n k} \varphi^*_k e^{i\omega_k t}],
    \qquad
    \mathbb{A}_{2,2n}(t) = \mathbb{A}^*_{1,2n}(t).
\end{gather}
With this knowledge, we can approximate $M(t)$ in Eq.~\eqref{SM:eq:M} on the time grid using a composite trapezoidal rule.
Finally, we vectorize the Lyapunov equation~\eqref{SM:eq:lyapunov} and we solve it using a standard Runge-Kutta integration routine.

Once $\Gamma$ has been obtained, real-space quantities can be calculated by inverting the transformation to Majorana operators.
For example, the population of the $n$th mode (which is what we used in the Letter) is found to be (remember that $\Gamma$ is pure imaginary):
\begin{equation}
    \langle a_n^\dag a_n \rangle = \frac{1}{2}\qty(1 + i \Gamma_{2n-1,2n}).
\end{equation}

\end{document}